\documentclass[onecolumn,12pt,draftclsnofoot]{IEEEtran}

\usepackage{mathrsfs}
\usepackage{mathptmx}
\usepackage{amsmath}
\usepackage{amsfonts}
\usepackage{amssymb}
\usepackage{graphicx}
\usepackage{xcolor}
\usepackage{float}
\usepackage{subfigure}
\usepackage{psfrag}
\usepackage{bm}
\usepackage{cite}
\usepackage{booktabs}
\newtheorem{defi}{Definition}
\newtheorem{thm}{Theorem}

\newtheorem{cor}{Corollary}
\newtheorem{rem}{Remark}

\newtheorem{lem}{Lemma}
\def\tp{\mathrm{T}}
\def\sn{\mathrm{span}}
\def\nl{\mathrm{null}}
\def\rank{\mathrm{rank}}

\renewcommand\baselinestretch{1.3}

\begin{document}

\title{From Control to Mathematics--Part I: Controllability-Based Design for Iterative Methods in Solving Linear Equations}

\author{Deyuan Meng, {\it Senior Member}, {\it IEEE}, and Yuxin Wu 

\thanks{This work was supported by the National Natural Science Foundation of China under Grant 61873013 and Grant 61922007.}
\thanks{The authors are with the Seventh Research Division, Beihang University (BUAA), Beijing 100191, P. R. China, and also with the School of Automation Science and Electrical Engineering, Beihang University (BUAA), Beijing 100191, P. R. China (e-mail: dymeng@buaa.edu.cn).}
}

\date{}
\maketitle

\begin{abstract}
In the interaction between control and mathematics, mathematical tools are fundamental for all the control methods, but it is unclear how control impacts mathematics. This is the first part of our paper that attempts to give an answer with focus on solving linear algebraic equations (LAEs) from the perspective of systems and control, where it mainly introduces the controllability-based design results. By proposing an iterative method that integrates a learning control mechanism, a class of tracking problems for iterative learning control (ILC) is explored for the problem solving of LAEs. A trackability property of ILC is newly developed, by which analysis and synthesis results are established to disclose the equivalence between the solvability of LAEs and the controllability of discrete control systems. Hence, LAEs can be solved by equivalently achieving the perfect tracking tasks of resulting ILC systems via the classic state feedback-based design and analysis methods. It is shown that the solutions for any solvable LAE can all be calculated with different selections of the initial input. Moreover, the presented ILC method is applicable to determining all the least squares solutions of any unsolvable LAE. In particular, a deadbeat design is incorporated to ILC such that the solving of LAEs can be completed within finite iteration steps. The trackability property is also generalized to conventional two-dimensional ILC systems, which creates feedback-based methods, instead of the common used contraction mapping-based methods, for the design and convergence analysis of ILC.
\end{abstract}


\section{Introduction}\label{sec1}

\IEEEPARstart{S}{olving} linear algebraic equations (LAEs) has been one of the most fundamentally significant problems in science and engineering \cite{mc:05,sv:00}. It has caught attention in lots of areas since many practical and complex problems can ultimately boil down to the solving of LAEs, such as forecast, estimation, and approximation of nonlinear systems and modelling of physical systems. In general, there are two basic categories of methods for solving LAEs, that is, direct methods and iterative methods. The direct methods are traditionally adopted for solving LAEs, which leverage the Gaussian elimination and its enhancements to directly calculate the exact solutions of LAEs with a process of finite steps. Though the direct methods can lead to the exact solutions of LAEs, it is lack of robustness and requires a large amount of computation and storage, especially for LAEs with higher dimensions. The iterative methods mainly make use of the successive iterations to update the approximate solutions of LAEs such that the exact solutions of them can be determined when the convergence of iterative process is achieved. Thanks to that the iterative methods only need to perform the addition and multiplication operations over a few matrices and vectors, they are generally preferred in practical applications \cite{k:95,g:97}.
%

The classic iterative methods for solving LAEs make certain use of the feedback mechanism. However, they ignore the use of ``control design'' in feedback, due to which the convergence rate, or even the convergence, for them may not be guaranteed. What will emerge if we can incorporate the idea of control into the iterative methods, and can this make them more effective in solving LAEs? If the answers to them are affirmative, how can we reasonably integrate control into the iterative methods? For this key problem, a promising control-theoretic design method is explored in \cite{bk:03,bk:07}, which not only presents continuous-time and discrete-time algorithms but also is effective, regardless of linear or nonlinear LAEs. In \cite{hj:05,hjl:06}, the design methods from optimal control and robust control are successfully introduced to solve LAEs, which yield solution algorithms with the global convergence or tunable convergence properties. Recently, there have been reported insightful consensus-based network control methods for solving LAEs, together with providing distributed algorithms based on local information exchange (see, e.g., \cite{mlm:15,lmnb:17,ae:20,zcf:20}). It has been revealed in \cite{bk:03,bk:07,hj:05,hjl:06,mlm:15,lmnb:17,ae:20,zcf:20} that the introduction of control design into iterative methods can make them to possess better performances in the crucial aspects of, e.g., complexity, convergence, and robustness, and can also provide a systematic way to design solution algorithms to LAEs.

However, the use of control methods to better performances of iterative methods is less developed especially in comparison with the popular application of iterative methods in the control area (see also \cite{c:01} for similar discussions made on the relation between control design and numerical analysis). Though some attempts are devoted to concerning this issue in, e.g., \cite{bk:03,bk:07,hj:05,hjl:06,mlm:15,lmnb:17,ae:20,zcf:20}, these existing results are mainly aimed at solvable LAEs (with a unique or at least one solution). It is required to further study whether and how the control design ideas can be leveraged to determine all solutions (or least squares solutions) for solvable (or unsolvable) LAEs. More fundamentally, are there inherent relationships between the solvability or unsolvability problems for LAEs and the basic properties for control systems, such as controllability, reachability, or stabilizability? Furthermore, in case of a positive answer to this question, how can the inherent relationships be reasonably disclosed? To our knowledge, these problems still remain open in the interaction between iterative methods for solving LAEs and control design methods, which will be addressed in this paper. Our approach is to incorporate the strategies of iterative learning control (ILC) in establishing effective solving methods for LAEs, in which a salient idea of learning from experience can be leveraged to connect iterative methods with control systems.

As a class of intelligent control methods, ILC contributes to realizing the ``perfect output tracking'' of any desired reference, which is accomplished by learning information from previous iterations and iteratively updating the control input signal (see, e.g., \cite{bta:06,acm:07,x:11}). Once the convergence for the iterative process is ensured, together with decreasing the tracking error to zero, the output tracking objective for any specified ILC system can be achieved perfectly over a fixed interval from the beginning to the end. It has been reported in the literature (see, e.g., \cite{xh:09,hx:09}) that ILC can be enabled to improve the transient response performances of the controlled systems, and be implemented in an off-line manner based on the input and output data, together with using quite limited model knowledge. Simultaneously, the update of the control input for ILC arrives at the approximation to the desired input for the desired reference, which is actually consistent with the exploration of iterative methods in realizing the approximation to the solution for an LAE. This consistency result brings the possibility to solve LAEs from the perspective of control analysis and synthesis based on constructing certain ILC systems, which however has not been investigated in the literature to the best of our knowledge.

With appropriate design of iterative learning controllers, the convergence performance of iterative processes may be greatly enhanced. The design and analysis of ILC mainly resorts to the contraction mapping method, rendering ILC distinct especially from the popular feedback-based control methods and separate from them simultaneously, e.g., see \cite{sw:03,tx:03,kcb:05,lxh:14,chjwc:15,xypy:16,awob:17,cdeag:17,mm:20}. It is known that the controllability is considered as a fundamentally significant property in most popular control methods, whereas it is seldom used in ILC, and instead, the system relative degree is regarded as one of the essential conditions in ILC. Are basic properties, such as reachability, controllability, or stabilizability, of control systems really not significant for ILC? The answer is negative. For the design and analysis of ILC, a realizability hypothesis is usually imposed, which embeds the controllability information of ILC systems (see, e.g., \cite{lgn:19}). However, where and how do these properties take effect in ILC? Furthermore, how can we benefit from them to proceed the studies of ILC and establish a close connection of ILC to the popular state feedback-based control methods? These questions still remain unanswered.

In this paper, we provide a viewpoint of systems and control to solve LAEs by constructing a framework of ILC, for which a fundamental trackability property of ILC is newly introduced by requiring only the existence of some desired inputs for ILC systems in generating the desired reference. It helps to develop the equivalent relation between the solvability of LAEs and the controllability of control systems, as shown in Fig. \ref{p1}. Further, we can arrive at an ILC algorithm to calculate the solutions of LAEs by leveraging the state feedback-based control methods. In addition, we explore the trackability property for traditional two-dimensional (2-D) ILC systems, and thus can benefit from the controllability and state feedback-based design methods to implement the perfect tracking tasks of ILC. Specifically, three main contributions of this paper are summarized as follows.
\begin{enumerate}
\item
We can incorporate the state feedback-based control into the solving of LAEs. It is thanks to the transformation of this problem to a tracking problem of ILC, which can be further disclosed with a connection to the controllability property. Thus, we can give an algorithm to determine all solutions (respectively, all least squares solutions) of any solvable (respectively, unsolvable) LAE by the selections of different initial inputs. By comparison, we may extend the existing control-theoretic iterative methods given in, e.g., \cite{bk:03,bk:07,hj:05,hjl:06,mlm:15,lmnb:17,ae:20,zcf:20} by incorporating the idea of ILC to derive all (least squares) solutions for LAEs.

\item
We propose a design method to improve the convergence rate of the iterative processes that result from the solving of LAEs. By integrating the idea of deadbeat control into the design of ILC, we can enable the convergence of the iterative processes to be accomplished in finite iteration steps. This not only overcomes the shortcoming of slow convergence in the classic iterative methods for solving LAEs, but also preserves their advantages of robustness and simple calculation. It thus renders iterative methods more practical for determining the solutions to LAEs.

\item
We disclose that the realizability hypothesis, required in many ILC literature (see, e.g., \cite{sw:02,s:05,mm:171,sx:20}), is unnecessary for the perfect tracking of ILC. For any trackable desired reference, there generally allow multiple desired inputs that generate this desired reference. We also extend the trackability property to the traditional 2-D ILC systems, and hence the state feedback-based design and analysis methods are effective in implementing the perfect output tracking tasks for them, which narrows the gap between ILC and the classic feedback-based control methods.
\end{enumerate}

As a dual result of this paper, observability-based design results for solving LAEs will be introduced in \cite{mw:21}.

\begin{figure}
\centering
\includegraphics[width=3.2in]{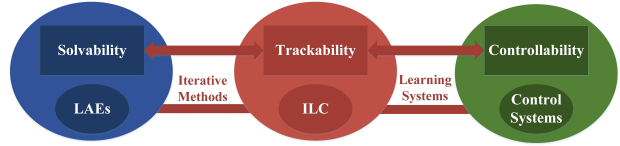}\\
\caption{The equivalent relationship between the solvability of LAEs and the controllability of control systems built via the trackability of ILC.}\label{p1}
\end{figure}

The rest of this paper is organized as follows. We establish a close connection between the solving problem of LAEs and the tracking problem of ILC in Section \ref{sec2}. In Sections \ref{sec3} and \ref{sec4}, the trackability and the controllability of ILC are introduced, respectively, and their relation is further disclosed and used to realize the perfect tracking of ILC based on the state feedback-based design and analysis methods in Section \ref{sec5}. In Section \ref{sec9}, we incorporate the idea of deadbeat control into the design for ILC such that we can develop the convergence for ILC within finite iterations. In Section \ref{sec6}, we present two implementation algorithms of ILC for solving LAEs, regardless of whether the LAEs are solvable. As an application, we apply the trackability property to traditional 2-D ILC systems and achieve the perfect tracking tasks for them with the state feedback-based methods in Section \ref{sec7}. We make concluding remarks in Section \ref{sec8}.

{\it Notations:} Let $\mathbb{Z}_+\triangleq\{0, 1, 2, \cdots\}$, $\mathbb{Z}_N\triangleq\{0, 1, \cdots, N \}$, and $I$ be an identity matrix with required dimensions. For any vector $x\in\mathbb{R}^{n}$, $\left\|x\right\|$ is its any norm, and $\left\|x\right\|_{2}$ and $\left\|x\right\|_{\infty}$ are particularly the Euclidean and infinite norms of $x$, respectively. Let $\rho(A)$ be the spectral radius of any matrix $A\in\mathbb{R}^{n\times n}$. For any linear space $X\subseteq\mathbb{R}^{n}$, $\dim(X)$ denotes its dimension. Let $X^{\bot}$ be the orthogonal complement subspace of any linear subspace $X$ in $\mathbb{R}^{n}$, and thus the direct sum of them is the entire space, which is denoted by $\mathbb{R}^{n}=X\oplus X^{\bot}$. For any vector $z_{k}\in\mathbb{R}^{n}$ changing with respect to an iteration axis that is denoted by $k\in\mathbb{Z}_{+}$, let $\Delta:z_{k}\to \Delta z_{k}\triangleq z_{k+1}-z_{k}$ represent a forward iteration operator.

\section{Problem-Solving of LAEs Via ILC Tracking}\label{sec2}

Let any desired reference $\bm{Y}_{d}\in\mathbb{R}^{p}$ be specified. The problem of interest in this paper is to solve an LAE expressed by
\begin{equation}\label{eq1}
\bm{Y}_{d}=\bm{G}\bm{U}_{d}
\end{equation}

\noindent to obtain the solution $\bm{U}_{d}\in\mathbb{R}^{q}$, where $\bm{G}\in\mathbb{R}^{p\times q}$ is the transfer or mapping matrix. If the LAE (\ref{eq1}) has solutions, then it is said to be {\it solvable}. Generally, there exist multiple solutions for the LAE (\ref{eq1}) when it is solvable. We thus need to not only calculate a certain solution but also determine an analytical formulation of the multiple solutions for the LAE (\ref{eq1}).

We target at addressing the solving problem of the LAE (\ref{eq1}) from the perspective of ILC. Toward this end, we duplicate the LAE (\ref{eq1}) in an iterative manner to arrive at a process, described with respect to the iteration index $k\in\mathbb{Z}_{+}$, as
\begin{equation}\label{eq2}
\bm{Y}_{k}=\bm{G}\bm{U}_{k},\quad\forall k\in\mathbb{Z}_{+}
\end{equation}

\noindent where we call $\bm{Y}_{k}\in\mathbb{R}^{p}$ and $\bm{U}_{k}\in\mathbb{R}^{q}$ the output and the input of (\ref{eq2}), respectively. Since (\ref{eq2}) can be employed for the description of an ILC system implicitly involving an independent discrete-time dynamics (see, e.g., \cite{bta:06,acm:07,m:19} for more discussions), it is also directly called a system for clarity, in accordance with which $\bm{U}_{d}$ is called the desired input for generating the desired reference $\bm{Y}_{d}$.

With the construction of the system (\ref{eq2}), the solving problem of the LAE (\ref{eq1}) can be interpreted as an ILC tracking problem.

{\it Problem Statement:} For the problem-solving of the LAE (\ref{eq1}), our objective is to design an input sequence $\left\{\bm{U}_{k}:k\in\mathbb{Z}_{+}\right\}$ for the system (\ref{eq2}), with any initial input $\bm{U}_{0}$ that can be arbitrarily selected, such that the resulting output sequence $\left\{\bm{Y}_{k}:k\in\mathbb{Z}_{+}\right\}$ can approach the desired reference $\bm{Y}_{d}$ as $k\to\infty$, namely,
\begin{equation}\label{eq3}
\lim_{k\to\infty}\bm{Y}_{k}=\bm{Y}_{d}
\end{equation}

\noindent and we can also learn a desired input $\bm{U}_{d}$ through the design of $\bm{U}_{k}$, $\forall k\in\mathbb{Z}_{+}$, and identify the unified properties of all possible desired inputs simultaneously after the ILC process.

To address the abovementioned tracking problem in ILC, we denote the tracking error as $\bm{E}_{k}=\bm{Y}_{d}-\bm{Y}_{k}$, $\forall k\in\mathbb{Z}_{+}$. Obviously, (\ref{eq3}) holds if and only if $\lim_{k\to\infty}\bm{E}_{k}=0$. Let the output deviation and input deviation between two sequential iterations for \eqref{eq2} be denoted as $\Delta\bm{Y}_{k}=\bm{Y}_{k+1}-\bm{Y}_{k}$, $\forall k\in\mathbb{Z}_{+}$ and $\Delta\bm{U}_{k}=\bm{U}_{k+1}-\bm{U}_{k}$, $\forall k\in\mathbb{Z}_{+}$, respectively. Then the use of (\ref{eq2}) yields
\begin{equation}\label{eq4}
\bm{E}_{k+1}
=\bm{E}_{k}-\Delta\bm{Y}_{k}
=\bm{E}_{k}+\bm{G}\bm{\Psi}_{k},\quad\forall k\in\mathbb{Z}_{+}
\end{equation}

\noindent which describes a dynamic system evolving along the iteration axis $k$ and having a control input given by
\begin{equation}\label{eq5}
\bm{\Psi}_{k}
=-\Delta\bm{U}_{k},\quad\forall k\in\mathbb{Z}_{+}.
\end{equation}

For distinction from the notations adopted in the system (\ref{eq2}), we present the following definition to introduce some notations of the linear discrete control system (\ref{eq4}).

\begin{defi}\label{defi1}
For the system (\ref{eq4}), $\bm{E}_{k}$ and $\bm{\Psi}_{k}$ are called the $k$-state and the $k$-input, respectively. This system, controlled with some certain $k$-state feedbacks, is said to be $k$-stable if, for any initial $k$-state condition, $\lim_{k\to\infty}\bm{E}_{k}=0$ can be accomplished.
\end{defi}

Based on Definition \ref{defi1}, the tracking objective (\ref{eq3}) is achieved for the system (\ref{eq2}) if and only if the $k$-stability is accomplished for the system (\ref{eq4}). As a consequence of this equivalence result, the design of the input $\bm{U}_{k}$ for the system (\ref{eq2}) to get the tracking objective (\ref{eq3}) can be transformed into designing the $k$-input $\bm{\Psi}_{k}$ for the system (\ref{eq4}) to realize the $k$-stability. It can be easily seen that the $k$-state is accessible for the feedback controller design of the system (\ref{eq4}). This renders the classic state feedback-based design tools in the Kalman state-space framework available for the ILC design of the system (\ref{eq2}) in the presence of the output tracking task (\ref{eq3}), and consequently available for the problem-solving of the LAE (\ref{eq1}). Motivated by such an observation, we will further develop a {\it trackability} property for the system \eqref{eq2}, which matches exactly with the {\it solvability} of the LAE \eqref{eq1}, and disclose the close relation between them and the {\it controllability} of the system \eqref{eq4} (see also Fig. \ref{p1}) such that the solutions of the LAE \eqref{eq1} can be obtained by a $k$-state feedback-based method.

\section{Trackability in ILC}\label{sec3}

In this section, we propose a fundamental class of trackability problems for the desired references in realizing the tracking tasks of ILC. We introduce the basic trackability properties and criteria, and disclose tight relations of them with the controlled systems of ILC at the same time.

\subsection{Trackability and Trackability Subspace}

Typically, existing ILC results aim at how to design updating laws of inputs to achieve the output tracking objectives and at how to analyze the convergence performances of their resulting ILC processes. However, it is lack of investigations for ILC to explore whether the desired references are trackable with some certain inputs. To address this crucial problem, the realizability assumption is often made directly for the desired references in the ILC literature, but its reasonability has not been developed. For the sake of our discussions on these fundamental problems of ILC, we introduce the following concepts of trackability and realizability for the desired references in ILC, and also present the related properties of the corresponding ILC systems.

\begin{defi}\label{defi2}
A desired reference $\bm{Y}_{d}$ is said to be trackable (respectively, realizable) for the system (\ref{eq2}) if there exists some (respectively, a unique) desired input $\bm{U}_{d}$ fulfilling the LAE (\ref{eq1}).
\end{defi}

\begin{defi}\label{defi4}
The system \eqref{eq2} is said to have the trackability property (respectively, the realizability property) if any desired reference $\bm{Y}_{d}\in\mathbb{R}^p$ is trackable (respectively, realizable) for the system \eqref{eq2}.
\end{defi}

With Definitions \ref{defi2} and \ref{defi4}, we connect the solvability problem of the LAE (\ref{eq1}) tightly to two basic problems of the trackability and the realizability in ILC for the system (\ref{eq2}).

\begin{rem}\label{rem1}
From Definition \ref{defi2}, it is clear that the trackability and realizability represent different concepts for characterizing the output tracking tasks of ILC systems. The realizability of a desired reference requires the uniqueness of the corresponding desired input for the controlled system, but by comparison, the trackability does not. In the literature of ILC, the realizability is usually adopted as a basic assumption, especially for under-actuated systems with less number of the inputs than that of the outputs (see, e.g., \cite{s:05,mm:171,sx:20}). Despite this fact, the trackability provides a more general property than the realizability for ILC by removing the uniqueness requirement, through which more fundamentally important problems in ILC may be explored.
\end{rem}
\begin{rem}\label{12}
From Definition \ref{defi4}, we know that an ILC system has the trackability (respectively, realizability) property if any desired reference is trackable (respectively, realizable). We can clearly see that these two properties for an ILC system are two inherent characteristics independent from the selections of the desired references. When the ILC system dose not possess the trackability property, only a portion of the desired references may be trackable, for which the corresponding desired inputs are available; and otherwise, it is impossible to determine any desired input by iterations. Thus, it is a crucial issue to identify whether the given desired reference is trackable before seeking control inputs for the system \eqref{eq2}.
\end{rem}

It is worth highlighting, however, that few results of ILC are devoted to exploring whether any desired reference of interest is trackable for the system \eqref{eq2}. If not, what properties do the trackable desired references have, and how do they connect to the characteristics of the system (\ref{eq2})? Furthermore, whether and how could some easy-to-check criteria be developed to validate the trackability of the desired references of interest? To address these fundamental problems, we focus on exploiting the output trackability of ILC for the system (\ref{eq2}) during accomplishing the tracking objective (\ref{eq3}) in the presence of any desired reference, and on characterizing the roles of the trackability of the desired references in performing the tracking tasks of ILC.

In Definition \ref{defi2}, the satisfaction of the LAE (\ref{eq1}) can guarantee that the linear combination of two trackable desired references for the system (\ref{eq2}) is also trackable. This evidently hints that all trackable desired references for the system (\ref{eq2}) span a subspace of $\mathbb{R}^{p}$. Let us denote this subspace as $\mathcal{Y}_{T}\subseteq\mathbb{R}^{p}$ that is directly called the {\it trackability subspace} of the system (\ref{eq2}) for simplicity and adopt $\bm{Y}_{d}\in\mathcal{Y}_{T}$ to denote any trackable $\bm{Y}_{d}$ in $\mathbb{R}^{p}$. It clearly follows from Definition \ref{defi4} that $\mathcal{Y}_{T}=\mathbb{R}^{p}$ holds if and only if the system \eqref{eq2} has the trackability property. However, by contrast, similar properties can not be analogously shown for the system (\ref{eq2}) when it is concerned with the realizable desired references owing to the requirement of uniqueness. This problem depends heavily upon the column independence of the matrix $\bm{G}$, which will be addressed with details in the following subsection.

\subsection{Trackability Criteria}

Let $\sn\bm{G}$ be the space that is spanned by the columns of $\bm{G}$. Let us also denote $\dim\left(\sn\bm{G}\right)=m$. Then we represent a set of basis vectors for $\sn\bm{G}$, defined by $\left\{h_{i}\in\mathbb{R}^{p}:i=1,2,\cdots,m\right\}$, with which we denote $\bm{H}_{1}=\left[h_{1},h_{2},\cdots,h_{m}\right]\in\mathbb{R}^{p\times m}$. Moreover, we choose some matrix $\bm{H}_{2}=\left[h_{m+1},h_{m+2},\cdots,h_{p}\right]\in\mathbb{R}^{p\times(p-m)}$ such that $\bm{H}=\left[\bm{H}_{1}~\bm{H}_{2}\right]\in\mathbb{R}^{p\times p}$ is nonsingular. Correspondingly, we denote $\bm{H}^{-1}=\bm{F}$, and let $\bm{F}$ be in a structured form of
\[
\bm{F}=\begin{bmatrix}\bm{F}_{1}^{\tp}\\\bm{F}_{2}^{\tp}\end{bmatrix}
~\hbox{with}~\bm{F}_{1}\in\mathbb{R}^{p\times m},\bm{F}_{2}\in\mathbb{R}^{p\times(p-m)}.
\]

Clearly, $\bm{H}_{1}$, $\bm{H}_{2}$, $\bm{F}_{1}$, and $\bm{F}_{2}$ are full-column rank matrices. These denotations also help develop a useful lemma as follows.

\begin{lem}\label{lem01}
For any matrix $\bm{G}$, the following statements hold:
\begin{enumerate}
\item
$\left(\sn\bm{G}\right)^{\bot}=\sn\bm{F}_{2}$ is such that $\sn\bm{G}\oplus\sn\bm{F}_{2}=\mathbb{R}^{p}$;

\item
$\left(\bm{H}_{1}\bm{F}_{1}^{\tp}\right)^2=\bm{H}_{1}\bm{F}_{1}^{\tp}$ and $\left(\bm{H}_{2}\bm{F}_{2}^{\tp}\right)^2=\bm{H}_{2}\bm{F}_{2}^{\tp}$ are idempotent matrices such that $\bm{H}_{1}\bm{F}_{1}^{\tp}+\bm{H}_{2}\bm{F}_{2}^{\tp}=I$.
\end{enumerate}
\end{lem}

\begin{proof}
``1):'' From $\bm{F}\bm{H}=I$, $\bm{F}_{2}^{\tp}\bm{H}_{1}=0$ follows. Since the columns of $\bm{H}_{1}$ form a set of basis vectors for $\sn\bm{G}$, we know $\sn\bm{G}=\sn\bm{H}_{1}$. From these two facts, we can obtain
\begin{equation}\label{eq6}
\bm{\vartheta}^{\tp}\bm{\theta}=0,\quad\forall\bm{\theta}\in\sn\bm{G},\forall\bm{\vartheta}\in\sn\bm{F}_{2}.
\end{equation}

\noindent As a consequence of (\ref{eq6}), we can arrive at $\left(\sn\bm{G}\right)^{\bot}=\sn\bm{F}_{2}$, and then $\sn\bm{G}\oplus\sn\bm{F}_{2}=\mathbb{R}^{p}$ by considering $\dim\left(\sn\bm{G}\right)=m$ and $\dim\left(\sn\bm{F}_{2}\right)=p-m$.

``2):'' From $\bm{F}\bm{H}=I$, it is direct to get $\bm{F}_{1}^{\tp}\bm{H}_{1}=I$ and $\bm{F}_{2}^{\tp}\bm{H}_{2}=I$. This leads to that $\left(\bm{H}_{1}\bm{F}_{1}^{\tp}\right)^2=\bm{H}_{1}\bm{F}_{1}^{\tp}$ and $\left(\bm{H}_{2}\bm{F}_{2}^{\tp}\right)^2=\bm{H}_{2}\bm{F}_{2}^{\tp}$ are idempotent. Moreover, $\bm{H}_{1}\bm{F}_{1}^{\tp}+\bm{H}_{2}\bm{F}_{2}^{\tp}=I$ is a straightforward consequence of $\bm{H}\bm{F}=I$.
\end{proof}

To proceed with Lemma \ref{lem01}, we establish a trackability result of any desired reference in ILC.

\begin{lem}\label{lem02}
For any desired reference $\bm{Y}_{d}\in\mathbb{R}^{p}$ of the system (\ref{eq2}), the following statements are equivalent:
\begin{enumerate}
\item
$\bm{Y}_{d}\in\mathcal{Y}_{T}$;

\item
$\bm{Y}_{d}\in\sn\bm{G}$;

\item
$\bm{F}_{2}^{\tp}\bm{Y}_{d}=0$;

\item
$\bm{H}_{1}\bm{F}_{1}^{\tp}\bm{Y}_d=\bm{Y}_d$.
\end{enumerate}
\end{lem}

\begin{proof}
We prove the equivalences among four statements in a circular manner.

``1)$\Rightarrow$2):'' If $\bm{Y}_{d}\in\mathcal{Y}_{T}$, then the LAE (\ref{eq1}) holds for some $\bm{U}_{d}$ from Definition \ref{defi2}. Clearly, (\ref{eq1}) implies $\bm{Y}_{d}\in\sn\bm{G}$.

``2)$\Rightarrow$3):'' With $\left(\sn\bm{G}\right)^{\bot}=\sn\bm{F}_{2}$ in Lemma \ref{lem01}, $\bm{F}_{2}^{\tp}\bm{Y}_{d}=0$ is an immediate result of $\bm{Y}_{d}\in\sn\bm{G}$.

``3)$\Rightarrow$4):'' By $\bm{H}_{1}\bm{F}_{1}^{\tp}+\bm{H}_{2}\bm{F}_{2}^{\tp}=I$ in Lemma \ref{lem01} together with $\bm{F}_{2}^{\tp}\bm{Y}_{d}=0$, we can deduce
\[
\bm{H}_{1}\bm{F}_{1}^{\tp}\bm{Y}_d=\left(I-\bm{H}_{2}\bm{F}_{2}^{\tp}\right)\bm{Y}_d=\bm{Y}_d.
\]

``4)$\Rightarrow$1):'' Owing to $\sn\bm{G}=\sn\bm{H}_{1}$, $\bm{H}_{1}=\bm{G}\bm{H}_{3}$ holds for some $\bm{H}_{3}\in\mathbb{R}^{q\times m}$. We thus define $\bm{U}_{d}=\bm{H}_{3}\bm{F}_{1}^{\tp}\bm{Y}_{d}$ and it holds
\[
\bm{Y}_{d}=\bm{H}_{1}\bm{F}_{1}^{\tp}\bm{Y}_{d}=\bm{G}\left(\bm{H}_{3}\bm{F}_{1}^{\tp}\bm{Y}_{d}\right)=\bm{G}\bm{U}_{d}
\]

\noindent namely, the LAE (\ref{eq1}) holds. According to Definition \ref{defi2}, $\bm{Y}_{d}\in\mathcal{Y}_{T}$ follows immediately.
\end{proof}

By moving Lemma \ref{lem02} a bit further, we present a realizability result of any desired reference in ILC.

\begin{lem}\label{lem03}
Consider the system (\ref{eq2}) with any $\bm{Y}_{d}\in\mathcal{Y}_{T}$. Then $\bm{Y}_{d}$ is realizable if and only if $\rank\left(\bm{G}\right)=q$.
\end{lem}

\begin{proof}
For $\bm{Y}_{d}\in\mathcal{Y}_{T}$, let $\bm{U}_{d}^{1}$ and $\bm{U}_{d}^{2}$ be two desired inputs that both fulfill the LAE (\ref{eq1}), namely, $\bm{Y}_{d}=\bm{G}\bm{U}_{d}^{1}=\bm{G}\bm{U}_{d}^{2}$. This obviously leads to $\bm{G}\left(\bm{U}_{d}^{1}-\bm{U}_{d}^{2}\right)=0$. Since $\bm{U}_{d}^{1}$ and $\bm{U}_{d}^{2}$ denote any two desired inputs satisfying (\ref{eq1}), we can validate
\[
\bm{U}_{d}^{1}-\bm{U}_{d}^{2}=0
~\Leftrightarrow~\rank\left(\bm{G}\right)=q
\]

\noindent which, together with Definition \ref{defi2}, implies that $\bm{Y}_{d}$ is realizable if and only if $\rank\left(\bm{G}\right)=q$.
\end{proof}



With Lemma \ref{lem03}, we reveal that the realizability of the desired references in ILC of the system (\ref{eq2}) requires the matrix $\bm{G}$ being of full-column rank. Under the satisfaction of this requirement, we can obtain that the linear combination of any two realizable desired references for the system (\ref{eq2}) is also realizable. Thus, it follows that all realizable desired references for the system (\ref{eq2}) span a subspace of $\mathbb{R}^{p}$, denoted by $\mathcal{Y}_{R}\subseteq\mathbb{R}^{p}$. We directly call $\mathcal{Y}_{R}$ the {\it realizability subspace} of the system (\ref{eq2}) for convenience and use $\bm{Y}_{d}\in\mathcal{Y}_{R}$ to denote any realizable $\bm{Y}_{d}$ in $\mathbb{R}^{p}$. By Lemma \ref{lem03}, we particularly have $\mathcal{Y}_{R}\subseteq\mathcal{Y}_{T}\subseteq\mathbb{R}^{p}$.

Based on Lemmas \ref{lem02} and \ref{lem03}, the following theorem introduces basic trackability and realizability properties in ILC.

\begin{thm}\label{thm01}
For the system (\ref{eq2}), the following two properties can be developed.
\begin{enumerate}
\item
{\it Trackability:} $\mathcal{Y}_{T}=\sn\bm{G}$ holds, and in particular, $\mathcal{Y}_{T}=\mathbb{R}^{p}$ if and only if $\rank\left(\bm{G}\right)=p$.

\item
{\it Realizability:} Either $\mathcal{Y}_{R}=\mathcal{Y}_{T}$ or $\mathcal{Y}_{R}=\{0\}$ holds, where $\mathcal{Y}_{R}=\mathcal{Y}_{T}$ if and only if $\rank\left(\bm{G}\right)=q$.

\end{enumerate}
\end{thm}

\begin{proof}
``1):'' From the equivalent results between 1) and 2) in Lemma \ref{lem02}, $\mathcal{Y}_{T}=\sn\bm{G}$ follows directly. We, in particular, can arrive at
\[
\mathcal{Y}_{T}=\mathbb{R}^{p}\Leftrightarrow~\sn\bm{G}=\mathbb{R}^{p}\Leftrightarrow~\dim\left(\sn\bm{G}\right)=p\Leftrightarrow~\rank\left(\bm{G}\right)=p.
\]

\noindent Thus, the trackability result 1) is obtained.

``2):'' With Lemma \ref{lem03}, we can verify that $\mathcal{Y}_{R}=\mathcal{Y}_{T}$ if and only if $\rank\left(\bm{G}\right)=q$. Otherwise, let $\rank\left(\bm{G}\right)\neq q$. Then $\rank\left(\bm{G}\right)=m<q$ holds, and for the null space $\nl\bm{G}$ of $\bm{G}$, i.e.,
\[
\nl\bm{G}
=\left\{\bm{\phi}\in\mathbb{R}^{q}\big|\bm{G}\bm{\phi}=0\right\}
\]

\noindent we have $\nl\bm{G}=\left(\sn\bm{G}^{\tp}\right)^{\bot}$, and as a consequence,
\begin{equation}\label{eq7}
\dim\left(\nl\bm{G}\right)
=q-\dim\left(\sn\bm{G}^{\tp}\right)
=q-m>0.
\end{equation}

\noindent With (\ref{eq7}), there exist more than one desired inputs that generate every nonzero trackable desired reference based on Theorem 2 of \cite[Subchapter 3.10]{lt:85}. Namely, for every nonzero $\bm{Y}_{d}\in\mathcal{Y}_{T}$, we have $\bm{Y}_{d}\not\in\mathcal{Y}_{R}$, and thus $\mathcal{Y}_{R}=\{0\}$ holds owing to $\mathcal{Y}_{R}\subseteq\mathcal{Y}_{T}$. The realizability result 2) is developed.
\end{proof}

\begin{rem}\label{rem2}
By $\mathcal{Y}_{T}=\sn\bm{G}$, it discloses that the trackability subspace of the system (\ref{eq2}) is exactly the spanning space for the columns of $\bm{G}$. This, however, is generally no longer applicable for the realizability subspace $\mathcal{Y}_{R}$ of the system (\ref{eq2}), which only works (namely, $\mathcal{Y}_{R}=\sn\bm{G}$) if and only if $\bm{G}$ is of full-column rank. Otherwise, $\mathcal{Y}_{R}=\{0\}$ emerges, namely, there do not exist any nonzero realizable desired references. Therefore, we reveal with Theorem \ref{thm01} that the trackability for ILC systems is a much more available property than the realizability. Simultaneously, for the system (\ref{eq2}), the trackability and realizability properties can also be determined according to Theorem \ref{thm01}. To be specific, the system \eqref{eq2} has the trackability (respectively, realizability) property if and only if $\rank\left(\bm{G}\right)=p$ holds (respectively, $\bm{G}$ is an invertible square matrix).
\end{rem}

%
%

\section{Controllability in ILC}\label{sec4}

In this section, we introduce the fundamental controllability property and its related problems into the tracking tasks of ILC, which helps to bring a viewpoint of addressing the tracking problems of ILC, and thus of dealing with the problem-solving of LAEs, from the classic theories and methods of systems and control. We develop the ideas for controllability, controllability subspace, controllability criteria, and controllability decomposition of the linear discrete systems (see, e.g., \cite[Chapter 3]{am:06}) to ILC systems.

\subsection{Controllability and Controllability Subspace}


Next, we turn to explore the tracking error system (\ref{eq4}) that is formulated in the state-space form of a linear discrete system. To this end, let us introduce two concepts for the controllability of the linear discrete systems (see, e.g., \cite[Chapter 3]{am:06}).

\begin{defi}\label{defi5}
For the system \eqref{eq4}, a nonzero $k$-state $\bm{\beta}$ in $\mathbb{R}^{p}$ is said to be controllable if, for some finite $l\in\mathbb{Z}_+$, there exists some $k$-input sequence $\bm{\Psi}_k$, $\forall k\in\mathbb{Z}_{l-1}$ that transfers the $k$-state from $\bm{\beta}$ to the origin at the $l$th iteration.
\end{defi}

\begin{defi}\label{defi6}
For the system (\ref{eq4}), if all, some but not all, and none of the nonzero $k$-states in $\mathbb{R}^{p}$ are controllable, respectively, then (\ref{eq4}) is said to be completely controllable, incompletely controllable, and completely uncontrollable, respectively.
%
%
%
\end{defi}


With Definition \ref{defi5}, we can validate that all the controllable $k$-states of the system \eqref{eq4}, together with the null vector, constitute a subspace of $\mathbb{R}^{p}$. We call it the {\it controllability subspace} of the system \eqref{eq4}, which is denoted by $\mathcal{C}_{C}$. It follows from Definition \ref{defi6} that $\mathcal{C}_{C}=\mathbb{R}^p$ holds if and only if the system \eqref{eq4} is completely controllable. For clarity, let $\mathcal{C}_{NC}=\mathcal{C}_{C}^{\bot}$ be the {\it uncontrollability subspace} of the system (\ref{eq4}). We thus have $\mathcal{C}_{C}\oplus\mathcal{C}_{NC}=\mathbb{R}^{p}$, and can develop specific properties of both $\mathcal{C}_{C}$ and $\mathcal{C}_{NC}$ as follows.

\begin{thm}\label{thm03}
For the system (\ref{eq4}), the following controllability properties hold:
\begin{enumerate}
\item
$\mathcal{C}_{NC}=\nl\bm{G}^{\tp}$;

\item
$\mathcal{C}_{C}=\sn\bm{G}$.
\end{enumerate}
\end{thm}

\begin{proof}
``1):'' For any controllable $k$-state $\bm{\beta}\in\mathcal{C}_{C}$, there exist some $l\in\mathbb{Z}_{+}$ and the corresponding input $\left\{\bm{\Psi}_{k}:\forall k\in\mathbb{Z}_{l-1}\right\}$ such that the solution of the system (\ref{eq4}) starting with $\bm{\beta}$ arrives at the origin after $l$ iterations, namely,
\[
0=\bm{E}_{l}
=\bm{E}_{0}+\sum_{k=0}^{l-1}\bm{G}\bm{\Psi}_{k}
=\bm{\beta}+\sum_{k=0}^{l-1}\bm{G}\bm{\Psi}_{k}.
\]

\noindent This implies that $\bm{\beta}\in\mathcal{C}_{C}$ is a controllable $k$-state for the system (\ref{eq4}) if and only if it can be described in the form of
\begin{equation}\label{eq8}
\bm{\beta}
=-\sum_{k=0}^{l-1}\bm{G}\bm{\Psi}_{k},\quad\forall\bm{\beta}\in\mathcal{C}_{C}.
\end{equation}

``$\mathcal{C}_{NC}\subseteq\nl\bm{G}^{\tp}$:'' For any $\bm{\alpha}\in\mathcal{C}_{NC}$, we can obtain $\bm{\alpha}^{\tp}\bm{\beta}=0$, $\forall\bm{\beta}\in\mathcal{C}_{C}$. By specifically considering the controllable $k$-state generated by $\bm{\Psi}_{k}=-l^{-1}\bm{G}^{\tp}\bm{\alpha}$, we can resort to (\ref{eq8}) to derive
\begin{equation*}\label{}
0=\bm{\alpha}^{\tp}\bm{\beta}
=-\bm{\alpha}^{\tp}\sum_{k=0}^{l-1}\bm{G}\bm{\Psi}_{k}
=\bm{\alpha}^{\tp}\bm{G}\bm{G}^{\tp}\bm{\alpha}
\end{equation*}

\noindent which leads to $\bm{G}^{\tp}\bm{\alpha}=0$, i.e., $\bm{\alpha}\in\nl\bm{G}^{\tp}$. This clearly implies $\mathcal{C}_{NC}\subseteq\nl\bm{G}^{\tp}$.

``$\mathcal{C}_{NC}\supseteq\nl\bm{G}^{\tp}$:'' For any $\bm{\alpha}\in\nl\bm{G}^{\tp}$, $\bm{\alpha}^{\tp}\bm{G}=0$ holds, and thus we can leverage (\ref{eq8}) to arrive at
\begin{equation*}\label{}
\bm{\alpha}^{\tp}\bm{\beta}
=-\sum_{k=0}^{l-1}\left(\bm{\alpha}^{\tp}\bm{G}\right)\bm{\Psi}_{k}
=0,\quad\forall\bm{\beta}\in\mathcal{C}_{C}
\end{equation*}

\noindent which ensures $\bm{\alpha}\in\mathcal{C}_{NC}$. Consequently, $\mathcal{C}_{NC}\supseteq\nl\bm{G}^{\tp}$ holds.

By combining $\mathcal{C}_{NC}\subseteq\nl\bm{G}^{\tp}$ with $\mathcal{C}_{NC}\supseteq\nl\bm{G}^{\tp}$, the result 1) can be immediately obtained.

``2):'' For any $\bm{\alpha}\in\mathcal{C}_{NC}$, the use of $\mathcal{C}_{NC}=\nl\bm{G}^{\tp}$ in 1) yields $\bm{\alpha}^{\tp}\bm{\beta}=0$, $\forall\bm{\beta}\in\sn\bm{G}$. This implies $\sn\bm{G}\subseteq\mathcal{C}_{C}$. Again using $\mathcal{C}_{NC}=\nl\bm{G}^{\tp}$, we can derive
\begin{equation}\label{eq9}
\dim\left(\mathcal{C}_{NC}\right)+\dim\left(\sn\bm{G}\right)=p.
\end{equation}

\noindent From $\mathcal{C}_{C}\oplus\mathcal{C}_{NC}=\mathbb{R}^{p}$, it is clear that $\dim\left(\mathcal{C}_{NC}\right)+\dim\left(\mathcal{C}_{C}\right)=p$, which, together with (\ref{eq9}), yields
\[
\dim\left(\sn\bm{G}\right)
=\dim\left(\mathcal{C}_{C}\right)
=p-\dim\left(\mathcal{C}_{NC}\right).
\]

\noindent Thanks to $\sn\bm{G}\subseteq\mathcal{C}_{C}$, we actually have $\mathcal{C}_{C}=\sn\bm{G}$. Namely, the result 2) holds.
\end{proof}

We can see from Theorem \ref{thm03} that the controllability of any $k$-state depends only on the matrix $\bm{G}$. The good property benefits from the specific structure of the system (\ref{eq4}) that is established from the tracking task (\ref{eq3}) of ILC for the system (\ref{eq2}). Obviously, we can use Theorem \ref{thm03} and Definition \ref{defi6} to, respectively, obtain $\sn\bm{G}=\mathbb{R}^{p}$, $\{0\}\subset\sn\bm{G}\subset\mathbb{R}^{p}$, and $\bm{G}=0$ when (\ref{eq4}) is completely controllable, incompletely controllable, and completely uncontrollable, respectively. Next, the complete controllability of the system (\ref{eq4}) is directly called controllability for simplicity.

\subsection{Controllability Criteria}

With the property of the controllable $k$-states in Theorem \ref{thm03}, we can present the controllability criteria for the system (\ref{eq4}).

\begin{thm}\label{thm04}
For the system (\ref{eq4}), it is controllable if and only if any of the following conditions holds:
\begin{enumerate}
\item
(\ref{eq4}) is $k$-stabilizable;

\item
$\rank\left(\bm{G}\right)=p$.
\end{enumerate}
\end{thm}

\begin{proof}
If the system (\ref{eq4}) is controllable, then it naturally is $k$-stabilizable, and as a consequence, there exists some $k$-state feedback with a gain matrix $\bm{K}\in\mathbb{R}^{q\times p}$ in the form of
\begin{equation}\label{eq10}
\bm{\Psi}_{k}=-\bm{K}\bm{E}_{k},\quad\forall k\in\mathbb{Z}_{+}
\end{equation}

\noindent to stabilize the resulting closed-loop system, given by
\begin{equation}\label{eq11}
\bm{E}_{k+1}=\left(I-\bm{G}\bm{K}\right)\bm{E}_{k},\quad\forall k\in\mathbb{Z}_{+}.
\end{equation}

\noindent Clearly, the $k$-stability of the system (\ref{eq11}) holds if and only if
\begin{equation}\label{eq12}
\rho\left(I-\bm{G}\bm{K}\right)<1
\end{equation}

\noindent which implies the nonsingularity of $\bm{G}\bm{K}$. Hence, $\rank\left(\bm{G}\right)=p$ is immediate.

On the contrary, if $\rank\left(\bm{G}\right)=p$, then (\ref{eq12}) holds for some $\bm{K}\in\mathbb{R}^{q\times p}$, under which (\ref{eq11}) is a $k$-stable system. That is, the system (\ref{eq4}) is $k$-stabilizable under some $k$-state feedback controller (\ref{eq10}). Moreover, we incorporate the result of Theorem \ref{thm03} and can deduce
\[
\dim\left(\mathcal{C}_{C}\right)
=\dim\left(\sn\bm{G}\right)
=\rank\left(\bm{G}\right)
=p
\]

\noindent and hence, $\mathcal{C}_{C}=\mathbb{R}^{p}$. That is, the controllability of the system (\ref{eq4}) is obtained.
\end{proof}

As a counterpart result of Theorem \ref{thm04}, the following theorem is introduced in the presence of the incomplete controllability of the system (\ref{eq4}).

\begin{thm}\label{thm05}
If the system (\ref{eq4}) is incompletely controllable, then there exists some nonsingular linear transformation $\widehat{\bm{E}}_{k}=\bm{F}\bm{E}_{k}$ such that $\widehat{\bm{E}}_{k}^{C}=\bm{F}_{1}^{\tp}\bm{E}_{k}\in\mathbb{R}^{m}$ satisfies
\begin{equation}\label{eq13}
\widehat{\bm{E}}_{k+1}^{C}=\widehat{\bm{E}}_{k}^{C}+\bm{F}_{1}^{\tp}\bm{G}\bm{\Psi}_{k},\quad\forall k\in\mathbb{Z}_{+}
\end{equation}

\noindent and $\widehat{\bm{E}}_{k}^{NC}=\bm{F}_{2}^{\tp}\bm{E}_{k}\in\mathbb{R}^{p-m}$ satisfies
\begin{equation}\label{eq14}
\widehat{\bm{E}}_{k}^{NC}=\widehat{\bm{E}}_{0}^{NC},\quad\forall k\in\mathbb{Z}_{+}.
\end{equation}

\noindent Moreover, the subsystem (\ref{eq13}) is controllable.
\end{thm}

\begin{proof}
From (\ref{eq6}), we have $\bm{F}_{2}^{\tp}\bm{G}=0$. We thus apply $\widehat{\bm{E}}_{k}=\bm{F}\bm{E}_{k}=\left[\left(\widehat{\bm{E}}_{k}^{C}\right)^{\tp}~\left(\widehat{\bm{E}}_{k}^{NC}\right)^{\tp}\right]^{\tp}$ to the system (\ref{eq4}) and can derive
\begin{equation}\label{eq15}
\aligned
\begin{bmatrix}\widehat{\bm{E}}_{k+1}^{C}\\\widehat{\bm{E}}_{k+1}^{NC}\end{bmatrix}
&=\begin{bmatrix}\widehat{\bm{E}}_{k}^{C}\\\widehat{\bm{E}}_{k}^{NC}\end{bmatrix}
+\begin{bmatrix}\bm{F}_{1}^{\tp}\\\bm{F}_{2}^{\tp}\end{bmatrix}
\bm{G}\bm{\Psi}_{k}\\
&=\begin{bmatrix}\widehat{\bm{E}}_{k}^{C}\\\widehat{\bm{E}}_{k}^{NC}\end{bmatrix}
+\begin{bmatrix}\bm{F}_{1}^{\tp}\bm{G}\\0\end{bmatrix}\bm{\Psi}_{k},\quad\forall k\in\mathbb{Z}_{+}
\endaligned
\end{equation}

\noindent from which (\ref{eq13}) is straightforward. Also, (\ref{eq15}) implies
\[\widehat{\bm{E}}_{k+1}^{NC}
=\widehat{\bm{E}}_{k}^{NC}
,\quad\forall k\in\mathbb{Z}_{+}\]

\noindent which is equivalent to (\ref{eq14}).

For the subsystem (\ref{eq13}), we can further obtain $\rank\left(\bm{F}_{1}^{\tp}\bm{G}\right)=m$ because the use of the nonsingularity of $\bm{F}$ and $\rank\left(\bm{G}\right)=m$ leads to
\begin{equation}\label{e019}
\rank\left(\bm{G}\right)
=\rank\left(\bm{F}\bm{G}\right)
=\rank\left(\begin{bmatrix}\bm{F}_{1}^{\tp}\bm{G}\\0\end{bmatrix}\right)\\
=\rank\left(\bm{F}_{1}^{\tp}\bm{G}\right).
\end{equation}

\noindent Hence, (\ref{eq13}) is a controllable system by considering Theorem \ref{thm04} for this subsystem.
\end{proof}

\begin{rem}\label{rem4}
By Theorems \ref{thm03}, \ref{thm04} and \ref{thm05}, we obtain fundamental controllability properties for the tracking error systems of ILC. They actually provide criteria on the design of $k$-state feedback controllers to achieve the $k$-stability. Of particular note is the close relation between these results and rank conditions of the matrix $\bm{G}$. The full-row rank of $\bm{G}$ is a necessary and sufficient condition for the controllability of the system (\ref{eq4}). Otherwise, there exist uncontrollable $k$-states for the system (\ref{eq4}), of which a standard controllability decomposition is provided in (\ref{eq15}). In particular, it is interesting to derive (\ref{eq14}) by this decomposition, which reveals that all the uncontrollable components of the $k$-states are not dynamic but fixed along the iteration axis.
\end{rem}

\section{Controllability-Based Design of ILC}\label{sec5}

In this section, we connect the trackability problems for ILC systems to the fundamental controllability problems, based on which we explore the design and analysis of ILC in performing output tracking tasks. For this purpose, next we present a result to reveal how to verify the trackability of the desired references in ILC from the perspective of the controllability.

\begin{thm}\label{thm06}
For the system (\ref{eq2}) and the associated tracking error system (\ref{eq4}), it always holds:
\begin{equation}\label{eq16}
\mathcal{Y}_{T}=\mathcal{C}_{C}.
\end{equation}

\noindent Further, for any $\bm{Y}_{d}\in\mathbb{R}^{p}$, $\bm{Y}_{d}\in\mathcal{Y}_{T}$ if and only if
\begin{equation}\label{eq17}
\widehat{\bm{E}}_{k}^{NC}
=\bm{F}_{2}^{\tp}\bm{E}_{0}
=0,\quad\forall k\in\mathbb{Z}_{+}
\end{equation}

\noindent and otherwise, $\bm{Y}_{d}\not\in\mathcal{Y}_{T}$ if and only if
\begin{equation}\label{eq18}
\widehat{\bm{E}}_{k}^{NC}
=\bm{F}_{2}^{\tp}\bm{E}_{0}
\neq0,\quad\forall k\in\mathbb{Z}_{+}.
\end{equation}
\end{thm}

\begin{proof}
From Theorems \ref{thm01} and \ref{thm03}, we can directly gain (\ref{eq16}) because two subspaces are both identical with $\sn\bm{G}$. Further, we leverage (\ref{eq2}) and $\bm{F}_{2}^{\tp}\bm{G}=0$ by Lemma \ref{lem01} to validate that for any $\bm{Y}_{d}\in\mathbb{R}^{p}$,
\[\aligned
\widehat{\bm{E}}_{k}^{NC}
=0,\quad\forall k\in\mathbb{Z}_{+}
~&\Leftrightarrow~\bm{F}_{2}^{\tp}\left(\bm{Y}_{d}-\bm{G}\bm{U}_{k}\right)=0,\quad\forall k\in\mathbb{Z}_{+}\\
~&\Leftrightarrow~\bm{F}_{2}^{\tp}\bm{Y}_{d}=0\\
~&\Leftrightarrow~\bm{Y}_{d}\in\mathcal{Y}_{T}
\endaligned\]

\noindent where we also incorporate the equivalent results of Lemma \ref{lem02}. For the same reason, we can deduce
\[\aligned
\widehat{\bm{E}}_{k}^{NC}
\neq0,\quad\forall k\in\mathbb{Z}_{+}
~&\Leftrightarrow~\bm{F}_{2}^{\tp}\left(\bm{Y}_{d}-\bm{G}\bm{U}_{k}\right)\neq0,\quad\forall k\in\mathbb{Z}_{+}\\
~&\Leftrightarrow~\bm{F}_{2}^{\tp}\bm{Y}_{d}\neq0\\
~&\Leftrightarrow~\bm{Y}_{d}\not\in\mathcal{Y}_{T}.
\endaligned\]

\noindent Then all results of this theorem can be derived with (\ref{eq14}).
\end{proof}

Based on Theorem \ref{thm06}, we can study the trackability problems of ILC systems by instead treating the controllability problems of their resulting tracking error systems. This makes it feasible to address the trackability analysis and the updating law design of ILC by resorting only to its tracking error system. Further, it is worth noting that all trackable desired references correspond to $\widehat{\bm{E}}_{k}^{NC}=0$, $\forall k\in\mathbb{Z}_{+}$, namely, the uncontrollable components of their relevant $k$-states are always equal to zero.

To proceed with the development of Theorem \ref{thm06}, we develop a fundamental tracking result of ILC in the following theorem, regardless of the system (\ref{eq2}) in the presence of any plant model.

\begin{thm}\label{thm07}
Let any desired reference $\bm{Y}_{d}\in\mathbb{R}^{p}$ be specified. Then the system (\ref{eq2}) can realize the tracking objective (\ref{eq3}) under some updating law of ILC in the form of
\begin{equation}\label{eq19}
\bm{U}_{k+1}=\bm{U}_{k}+\bm{K}\bm{E}_{k},\quad\forall k\in\mathbb{Z}_{+}
\end{equation}

\noindent if and only if $\bm{Y}_{d}\in\mathcal{Y}_{T}$. Otherwise, the system (\ref{eq2}) can no longer achieve the tracking objective (\ref{eq3}), regardless of the application of any input sequence $\left\{\bm{U}_{k}:k\in\mathbb{Z}_{+}\right\}$. 
%
\end{thm}

\begin{proof}
From (\ref{eq5}), we can validate the equivalence between the updating law (\ref{eq19}) and the  $k$-state feedback (\ref{eq10}). Note that $\rank\left(\bm{F}_{1}^{\tp}\bm{G}\right)=m$ and $\rank\left(\bm{H}_{1}\right)=m$, namely, $\bm{F}_{1}^{\tp}\bm{G}$ and $\bm{H}_{1}$ are of full-row rank and of full-column rank, respectively. Hence, there always exists some gain matrix $\bm{K}\in\mathbb{R}^{q\times p}$ such that
\begin{equation}\label{eq20}
\rho\left(I-\bm{F}_{1}^{\tp}\bm{G}\bm{K}\bm{H}_{1}\right)<1
\end{equation}

\noindent and some gain matrix $\widehat{\bm{K}}\in\mathbb{R}^{q\times m}$ such that
\begin{equation}\label{eq21}
\rho\left(I-\bm{F}_{1}^{\tp}\bm{G}\widehat{\bm{K}}\right)<1.
\end{equation}

\noindent By these preparations, next we show that the tracking objective (\ref{eq3}) holds for the system (\ref{eq2}) under the updating law (\ref{eq19}) if and only if $\bm{Y}_{d}\in\mathcal{Y}_{T}$.

{\it Necessity:} If (\ref{eq3}) holds for (\ref{eq2}) under (\ref{eq19}), then $\lim_{k\to\infty}\bm{E}_{k}=0$ holds for (\ref{eq11}) (namely, the closed-loop system resulting from (\ref{eq4}) and (\ref{eq10})). With Lemma \ref{lem01} and from (\ref{eq11}), we can leverage $\widehat{\bm{E}}_{k}=\bm{F}\bm{E}_{k}$, $\widehat{\bm{E}}_{k}^{C}=\bm{F}_{1}^{\tp}\bm{E}_{k}$, and $\widehat{\bm{E}}_{k}^{NC}=\bm{F}_{2}^{\tp}\bm{E}_{k}$ to obtain (\ref{eq14}) and
\begin{equation}\label{eq22}
\widehat{\bm{E}}_{k+1}^{C}=\left(I-\bm{F}_{1}^{\tp}\bm{G}\bm{K}\bm{H}_{1}\right)\widehat{\bm{E}}_{k}^{C}
-\bm{F}_{1}^{\tp}\bm{G}\bm{K}\bm{H}_{2}\widehat{\bm{E}}_{k}^{NC},\quad\forall k\in\mathbb{Z}_{+}.
\end{equation}

\noindent Based on $\lim_{k\to\infty}\bm{E}_{k}=0$ and with $\widehat{\bm{E}}_{k}^{NC}=\bm{F}_{2}^{\tp}\bm{E}_{k}$, the use of (\ref{eq14}) leads to (\ref{eq17}). It hence follows from Theorem \ref{thm06} that $\bm{Y}_{d}\in\mathcal{Y}_{T}$. In addition, the substitution of (\ref{eq17}) into (\ref{eq22}) yields
\begin{equation}\label{e15}
\widehat{\bm{E}}_{k+1}^{C}=\left(I-\bm{F}_{1}^{\tp}\bm{G}\bm{K}\bm{H}_{1}\right)\widehat{\bm{E}}_{k}^{C},\quad\forall k\in\mathbb{Z}_{+}
\end{equation}

\noindent which is $k$-stable, and hence $\lim_{k\to\infty}\widehat{\bm{E}}_{k}^{C}=0$, under the spectral radius condition (\ref{eq20}). Clearly, this coincides with the two facts of $\lim_{k\to\infty}\bm{E}_{k}=0$ and of $\widehat{\bm{E}}_{k}^{C}=\bm{F}_{1}^{\tp}\bm{E}_{k}$.

{\it Sufficiency:} Let us consider the controllable subsystem (\ref{eq13}), and there always exists some $k$-state feedback $\bm{\Psi}_{k}=-\widehat{\bm{K}}\bm{E}_{k}^{C}$ to generate the closed-loop system given by
\begin{equation}\label{eq23}
\widehat{\bm{E}}_{k+1}^{C}=\left(I-\bm{F}_{1}^{\tp}\bm{G}\widehat{\bm{K}}\right)\widehat{\bm{E}}_{k}^{C},\quad\forall k\in\mathbb{Z}_{+}
\end{equation}

\noindent of which the $k$-stability can be ensured with the spectral radius condition (\ref{eq21}). This implies $\lim_{k\to\infty}\widehat{\bm{E}}_{k}^{C}=0$. If $\bm{Y}_{d}\in\mathcal{Y}_{T}$, then (\ref{eq17}) holds based on Theorem \ref{thm06}. Hence, we have $\lim_{k\to\infty}\widehat{\bm{E}}_{k}=0$ which, together with $\bm{E}_{k}=\bm{H}\widehat{\bm{E}}_{k}$, yields $\lim_{k\to\infty}\bm{E}_{k}=0$. Owing to $\widehat{\bm{E}}_{k}^{C}=\bm{F}_{1}^{\tp}\bm{E}_{k}$, we denote $\bm{K}=\widehat{\bm{K}}\bm{F}_{1}^{\tp}$, and as a consequence of $\bm{\Psi}_{k}=-\widehat{\bm{K}}\bm{E}_{k}^{C}$, we can deduce that (\ref{eq10}) is equivalent to (\ref{eq19}). By incorporating these facts, we can see that the tracking objective (\ref{eq3}) holds for the system (\ref{eq2}) under the updating law (\ref{eq19}) when designing its gain matrices of the form $\bm{K}=\widehat{\bm{K}}\bm{F}_{1}^{\tp}$ and choosing $\widehat{\bm{K}}$ according to (\ref{eq21}), regardless of any $\bm{Y}_{d}\in\mathcal{Y}_{T}$.

Otherwise, we can see from Theorem \ref{thm06} that $\bm{Y}_{d}\not\in\mathcal{Y}_{T}$ if and only if (\ref{eq18}) holds. It is immediate from (\ref{eq18}) that $\lim_{k\to\infty}\bm{E}_{k}=0$ does not hold any longer because of $\widehat{\bm{E}}_{k}^{NC}=\bm{F}_{2}^{\tp}\bm{E}_{k}$. Namely, (\ref{eq3}) can not be achieved for (\ref{eq2}), despite any $\bm{U}_{k}$, $\forall k\in\mathbb{Z}_{+}$.
\end{proof}

\begin{rem}\label{rem5}
In Theorem \ref{thm07}, we identify a property of ILC that the trackability of the desired references actually can provide a necessary and sufficient condition to realize the output tracking tasks. The property makes us able to provide ILC of the system (\ref{eq2}) with an output tracking result for a general case, regardless of any rank condition of $\bm{G}$. This is different from conventional ILC results that generally require a full (row or column) rank condition. In addition, Theorem \ref{thm07} introduces a way to leverage the state feedback-based design and analysis to address output tracking problems of ILC. We can thus achieve the design and analysis of ILC from the perspective of system stability, rather than resorting to traditional approaches for ILC based on, e.g., contraction mapping analysis.
\end{rem}
\begin{rem}\label{rem7}
With Theorem \ref{thm07}, it reveals that for any $\bm{Y}_{d}\in\mathcal{Y}_{T}$, we can design the updating law (\ref{eq19}) of ILC for the system (\ref{eq2}) to realize the tracking objective (\ref{eq3}). We can particularly design the gain matrix $\bm{K}$ of (\ref{eq19}) according to $\bm{K}=\widehat{\bm{K}}\bm{F}_{1}^{\tp}$, for which $\widehat{\bm{K}}$ is selected to satisfy (\ref{eq21}). Since $\bm{F}_{1}^{\tp}\bm{G}$ is of full-row rank, there always exists $\widehat{\bm{K}}$ that fulfills (\ref{eq21}). This spectral radius condition in fact ensures the $k$-stability of the system (\ref{eq23}) corresponding to the controllable components of the $k$-states. Because (\ref{eq19}) is actually a generalized form of the P-type updating law of ILC, Theorem \ref{thm07} also provides the explanation about why the P-type updating law is powerful and most applied in the ILC tracking tasks from the viewpoints of state feedback and controllability.
\end{rem}

In addition to the fundamental tracking problem of concern in Theorem \ref{thm07}, another fundamental problem emerging is: what is the performance of the correspondingly used input sequence, or what does ILC learn? To answer this question, we separately consider two cases depending on whether the desired reference is trackable for the system (\ref{eq2}) or not.

We first contribute to addressing the case for the system (\ref{eq2}) in the presence of the trackable desired references, and  denote the set of the desired inputs that generate any  $\bm{Y}_{d}\in\mathcal{Y}_{T}$ as
\begin{equation}\label{eq24}
\mathcal{U}_{d}\left(\bm{Y}_{d}\right)
=\left\{\bm{U}_{d}\in\mathbb{R}^{q}
\big|\bm{Y}_{d}=\bm{G}\bm{U}_{d}\right\},\quad\forall\bm{Y}_{d}\in\mathcal{Y}_{T}.
\end{equation}

\noindent It is clear that $\mathcal{U}_{d}\left(\bm{Y}_{d}\right)$ is exactly the set of the solutions to the LAE (\ref{eq1}) for any $\bm{Y}_{d}\in\mathcal{Y}_{T}$. We can further reveal that ILC can be designed to learn all the desired inputs collected in $\mathcal{U}_{d}\left(\bm{Y}_{d}\right)$ for any $\bm{Y}_{d}\in\mathcal{Y}_{T}$ through the selection of the initial inputs.

\begin{thm}\label{thm08}
Consider the system (\ref{eq2}) with any $\bm{Y}_{d}\in\mathcal{Y}_{T}$, and let the updating law (\ref{eq19}) be applied, and the condition (\ref{eq20}) be satisfied. Then for any initial input $\bm{U}_{0}$, the sequence of inputs $\bm{U}_{k}$, $\forall k\in\mathbb{Z}_{+}$, generated by (\ref{eq19}), converges such that the input learned with ILC forms a convex set as
\begin{equation}\label{eq25}
\aligned
\mathcal{U}_{\mathrm{ILC}}(\bm{Y}_{d})
=\bigg\{\bm{U}_{\infty}&=\left[I-\bm{K}\bm{H}_{1}\left(\bm{F}_{1}^{\tp}\bm{G}\bm{K}\bm{H}_{1}\right)^{-1}\bm{F}_{1}^{\tp}\bm{G}\right]\bm{U}_{0}\\
&~~~+\bm{K}\bm{H}_{1}\left(\bm{F}_{1}^{\tp}\bm{G}\bm{K}\bm{H}_{1}\right)^{-1}\bm{F}_{1}^{\tp}\bm{Y}_{d}
\Big|\bm{U}_{0}\in\mathbb{R}^{q}\bigg\},\\
&~~~~~~~~~~~~~~~~~~~~~~~~~~~~~~~~~~~~~~\forall\bm{Y}_{d}\in\mathcal{Y}_{T}
\endaligned
\end{equation}

\noindent where $\bm{U}_{\infty}\triangleq\lim_{k\to\infty}\bm{U}_{k}$. Moreover, it always holds
\begin{equation}\label{eq26}
\mathcal{U}_{\mathrm{ILC}}(\bm{Y}_{d})
=\mathcal{U}_{d}(\bm{Y}_{d}),\quad\forall\bm{Y}_{d}\in\mathcal{Y}_{T}
\end{equation}

\noindent where it particularly follows
\begin{equation}\label{eq27}
\mathcal{U}_{\mathrm{ILC}}(\bm{Y}_{d})
=\mathcal{U}_{d}(\bm{Y}_{d})
=\left\{\left(\bm{G}^{\tp}\bm{G}\right)^{-1}\bm{G}^{\tp}\bm{Y}_{d}\right\},\quad\forall\bm{Y}_{d}\in\mathcal{Y}_{T}
\end{equation}

\noindent if and only if $\mathcal{Y}_{T}=\mathcal{Y}_{R}$ holds. 
\end{thm}

\begin{proof}
See Appendix A.
\end{proof}

\begin{rem}\label{rem8}
We reveal by Theorem \ref{thm08}, together with Theorem \ref{thm07}, that for any trackable desired reference, we can leverage the design with a P-type ILC algorithm to not only accomplish the tracking objective but also determine a corresponding input to generate the desired reference. Furthermore, the inputs learned through ILC are dependent linearly upon the desired reference and the initial input, based on which we can develop all desired inputs associated with the desired reference with the selections of different initial inputs.
It is worth noting also for Theorems \ref{thm07} and \ref{thm08} that under the spectral radius condition (\ref{eq20}), we ensure their convergence results with an exponentially fast speed. 
\end{rem}%
%

When the full-row rank condition of $\bm{G}$ is ensured, Theorem \ref{thm07} can be developed to work for the tracking tasks of the system (\ref{eq2}) in the presence of any desired references.

\begin{cor}\label{cor01}
Let $\rank\left(\bm{G}\right)=p$ hold for the system (\ref{eq2}). Then for any desired reference $\bm{Y}_{d}\in\mathbb{R}^{p}$, there exists some updating law (\ref{eq19}) of ILC to achieve the tracking objective (\ref{eq3}), together with the resulting sequence of inputs converging into a convex set given by
\begin{equation}\label{eq38}
\aligned
\mathcal{U}_{\mathrm{ILC}}(\bm{Y}_{d})
=\bigg\{\bm{U}_{\infty}&=\left[I-\bm{K}\left(\bm{G}\bm{K}\right)^{-1}\bm{G}\right]\bm{U}_{0}\\
&~~~+\bm{K}\left(\bm{G}\bm{K}\right)^{-1}\bm{Y}_{d}
\Big|\bm{U}_{0}\in\mathbb{R}^{q}\bigg\},\quad\forall\bm{Y}_{d}\in\mathbb{R}^{p}
\endaligned
\end{equation}

\noindent for which the design condition of the gain matrix $\bm{K}$ is provided by (\ref{eq12}).
\end{cor}

\begin{proof}
Owing to $\rank\left(\bm{G}\right)=p$ and from Theorem \ref{thm04}, $\bm{H}_{1}=I$ and $\bm{F}_{1}=I$ can be directly taken. Then this corollary follows as an immediate consequence of Theorems \ref{thm07} and \ref{thm08}.
\end{proof}

As a counterpart of Corollary \ref{cor01}, the following tracking result of ILC under the full-column rank condition of $\bm{G}$ is presented.

\begin{cor}\label{cor02}
Let $\rank\left(\bm{G}\right)=q$ hold for the system (\ref{eq2}). Then for any realizable desired reference $\bm{Y}_{d}\in\mathcal{Y}_{R}$, there exists some updating law (\ref{eq19}) of ILC to both achieve the tracking objective (\ref{eq3}) and ensure
\begin{equation}\label{eq39}
\lim_{k\to\infty}\bm{U}_{k}=\left(\bm{G}^{\tp}\bm{G}\right)^{-1}\bm{G}^{\tp}\bm{Y}_{d},\quad\forall\bm{U}_{0}\in\mathbb{R}^{q}
\end{equation}

\noindent for which the design condition is to guarantee the gain matrix $\bm{K}$ to satisfy
\begin{equation}\label{eq40}
\rho\left(I-\bm{K}\bm{G}\right)<1.
\end{equation}
\end{cor}

\begin{proof}
For the system (\ref{eq2}) under the updating law (\ref{eq19}), the tracking result of (\ref{eq3}) follows straightforwardly from Theorem \ref{thm07}, regardless of any $\bm{Y}_{d}\in\mathcal{Y}_{R}$. Simultaneously, we use Theorem \ref{thm08} for the case $\rank\left(\bm{G}\right)=q$, and can obtain (\ref{eq39}) thanks to (\ref{eq27}), where $\bm{U}_{d}=\left(\bm{G}^{\tp}\bm{G}\right)^{-1}\bm{G}^{\tp}\bm{Y}_{d}$ is exactly the unique solution for the LAE (\ref{eq1}). Furthermore, if we denote $\delta\bm{U}_{k}=\bm{U}_{d}-\bm{U}_{k}$, then we can equivalently deduce from (\ref{eq1}) and (\ref{eq28}) that
\begin{equation}\label{eq41}
\delta\bm{U}_{k+1}=\left(I-\bm{K}\bm{G}\right)\delta\bm{U}_{k},\quad\forall k\in\mathbb{Z}_{+}.
\end{equation}

\noindent Note that we can ensure $\lim_{k\to\infty}\delta\bm{U}_{k}=0$ by (\ref{eq39}) for any initial condition $\delta\bm{U}_{0}$. By considering this stability result for (\ref{eq41}), the design condition (\ref{eq40}) of $\bm{K}$ can be developed straightforwardly. The proof of this corollary is completed.
\end{proof}

\begin{rem}\label{rem6}
In Corollaries \ref{cor01} and \ref{cor02}, we specifically contribute to ILC for the controlled systems with the full (row or column) rank conditions. This class of ILC problems is one of the most considered classes in the ILC literature, especially by applying the simple P-type updating law (\ref{eq19}) with the design conditions (\ref{eq12}) and (\ref{eq40}). Based on Corollaries \ref{cor01} and \ref{cor02}, we reveal that this specific class of ILC problems can be particularly addressed in our framework developed by benefiting from the controllability and trackability criteria related to ILC systems.
\end{rem}

From Theorem \ref{thm07}, if the desired reference $\bm{Y}_d$ is untrackable, that is, $\bm{Y}_d\notin\sn\bm{G}$, then it is impossible to determine an input sequence $\left\{\bm{U}_{k}:k\in\mathbb{Z}_{+}\right\}$ to realize the tracking objective \eqref{eq3} for the system \eqref{eq2}. In this circumstance, we contribute to exploring the updating law \eqref{eq19} for the system \eqref{eq2} such that the sequence of output $\left\{\bm{Y}_{k}:k\in\mathbb{Z}_{+}\right\}$ converges (namely, $\bm{Y}_{\infty}\triangleq\lim_{k\to\infty}\bm{Y}_{k}$ exists), and simultaneously produces a converged sequence of input $\left\{\bm{U}_k:k\in\mathbb{Z}_+\right\}$ able to minimize $\left\|\bm{Y}_{d}-\bm{G}\bm{U}_{\infty}\right\|_{2}$, namely,
\[
\left\|\bm{Y}_{d}-\bm{Y}_{\infty}\right\|_{2}
=\left\|\bm{Y}_{d}-\bm{G}\bm{U}_{\infty}\right\|_{2}
=\min_{\widetilde{\bm{U}}\in\mathbb{R}^{q}}\left\|\bm{Y}_{d}-\bm{G}\widetilde{\bm{U}}\right\|_{2},~\forall\bm{Y}_{d}\notin\mathcal{Y}_{T}.
\]

\noindent It obviously corresponds to the calculation of the least squares solutions for the LAE (\ref{eq1}). Towards this end, let us denote the set of all least squares solutions for the LAE (\ref{eq1}) in the presence of any untrackable desired reference $\bm{Y}_{d}\notin\mathcal{Y}_{T}$ as
\begin{equation}\label{e1}
\aligned
\overline{\mathcal{U}}_{d}\left(\bm{Y}_{d}\right)
=\left\{\overline{\bm{U}}_{d}\in\mathbb{R}^{q}
\Big|\left\|\bm{Y}_{d}-\bm{G}\overline{\bm{U}}_{d}\right\|_{2}
=\min_{\widetilde{\bm{U}}\in\mathbb{R}^{q}}\left\|\bm{Y}_{d}-\bm{G}\widetilde{\bm{U}}\right\|_{2}\right\}&,\\
\forall\bm{Y}_{d}\notin\mathcal{Y}_{T}&.
\endaligned
\end{equation}

For the tracking problem of the system \eqref{eq2} in the presence of any $\bm{Y}_{d}\notin\mathcal{Y}_{T}$, we can present the following theorem to develop learning results under the updating law \eqref{eq19}.

\begin{thm}\label{t1}
Consider the system (\ref{eq2}) with any $\bm{Y}_{d}\notin\mathcal{Y}_{T}$, and let the updating law (\ref{eq19}) be applied with $\bm{K}=\bm{\widehat{K}}\bm{F}^{\tp}_1$, and the condition (\ref{eq21}) be satisfied. Then for any initial input $\bm{U}_{0}$, the sequences of the output $\bm{Y}_{k}$, $\forall k\in\mathbb{Z}_{+}$ and the input $\bm{U}_{k}$, $\forall k\in\mathbb{Z}_{+}$ converge, respectively, to
\begin{equation}\label{e2}
\bm{Y}_{\infty}
=\bm{H}_{1}\bm{F}_{1}^{\tp}\bm{Y}_{d}
\end{equation}

\noindent and to $\bm{U}_{\infty}$ collected in the convex set as
\begin{equation}\label{e3}
\aligned
\overline{\mathcal{U}}_{\mathrm{ILC}}(\bm{Y}_{d})
=\Bigg\{\bm{U}_{\infty}&=\left[I-\widehat{\bm{K}}\left(\bm{F}_{1}^{\tp}\bm{G}\widehat{\bm{K}}\right)^{-1}\bm{F}_{1}^{\tp}\bm{G}\right]\bm{U}_{0}\\
&~~~+\widehat{\bm{K}}\left(\bm{F}_{1}^{\tp}\bm{G}\widehat{\bm{K}}\right)^{-1}\bm{F}_{1}^{\tp}\bm{Y}_{d}
\Big|\bm{U}_{0}\in\mathbb{R}^{q}\Bigg\},\\
&~~~~~~~~~~~~~~~~~~~~~~~~~~~~~~~\forall\bm{Y}_{d}\notin\mathcal{Y}_{T}.
\endaligned
\end{equation}

\noindent Moreover, it holds
\begin{equation}\label{e4}
\overline{\mathcal{U}}_{\mathrm{ILC}}(\bm{Y}_{d})
=\overline{\mathcal{U}}_{d}(\bm{Y}_{d}),\quad\forall\bm{Y}_{d}\notin\mathcal{Y}_{T}
\end{equation}

\noindent if and only if the selection of $\bm{H}$ is such that $\bm{H}_1^{\tp}\bm{H}_2=0$.
\end{thm}

\begin{proof}
See Appendix B.
\end{proof}

\begin{rem}\label{rem12}
In Theorem \ref{t1}, we reveal that we can still obtain convergent updating laws of ILC in the presence of untrackable desired references. Since $\bm{H}_{1}\bm{F}_{1}^{\tp}$ is idempotent, (\ref{e2}) shows that for any $\bm{Y}_{d}\notin\mathcal{Y}_{T}$, the output of the system (\ref{eq2}) converges not to $\bm{Y}_{d}$ but to its projection onto $\sn\left(\bm{H}_{1}\bm{F}_{1}^{\tp}\right)$ along $\sn\left(\bm{H}_{2}\bm{F}_{2}^{\tp}\right)$. In particular, we can minimize the residual tracking error in the sense of the least squares error by setting $\bm{H}$ such that $\bm{H}_{1}^{\tp}\bm{H}_{2}=0$, and correspondingly determine all least squares solutions for the LAE (\ref{eq1}) by resorting to the converged inputs of the system (\ref{eq2}) under different selections of the initial inputs. Since $\bm{H}_{2}$ is determined to arrive at a nonsingular matrix together with $\bm{H}_{1}$, this property and the condition $\bm{H}_{1}^{\tp}\bm{H}_{2}=0$ can be accomplished simultaneously for some $\bm{H}_{2}$. This actually renders it possible to always develop Theorem \ref{t1} for any selected matrix $\bm{H}_{1}$.
\end{rem}

\section{Finite-Iteration Convergence of ILC}\label{sec9}

In Theorems \ref{thm07}, \ref{thm08}, and \ref{t1}, the convergence of ILC only resorts to the spectral radius conditions (\ref{eq20}) and (\ref{eq21}) that can achieve the exponential convergence for the sequences of both $\bm{Y}_k$ and $\bm{U}_k$. A further question arising for the improvement of the ILC convergence rate is: whether and how can $\bm{Y}_k$ and $\bm{U}_k$ be driven to converge within finite iterations? If so, we can obtain control design tools to realize the convergence of iterative methods for solving LAEs within finite iterations, which works despite the calculation of solutions or least squares solutions for LAEs.

To provide the abovementioned question with an affirmative answer, we introduce the idea of deadbeat control to the design of ILC, and propose a finite iteration convergence result below.

\begin{thm}\label{t3}
Consider the system (\ref{eq2}) with any $\bm{Y}_{d}\in\mathbb{R}^{p}$, and let the updating law (\ref{eq19}) be applied. For $\bm{Y}_{d}\in\mathcal{Y}_{T}$, there exists some $\nu_{1}\in\mathbb{Z}_{m}$ ($m=\rank\left(\bm{G}\right)$) such that
\begin{equation}\label{e013}
\bm{U}_{k}
\left\{\aligned
&\neq\bm{U}_{d}, &\forall k&\leq\nu_{1}-1\\
&=\bm{U}_{d}, &\forall k&\geq\nu_{1}
\endaligned
\right.,
~~
\bm{Y}_{k}
\left\{\aligned
&\neq\bm{Y}_{d}, &\forall k&\leq\nu_{1}-1\\
&=\bm{Y}_{d}, &\forall k&\geq\nu_{1}
\endaligned
\right.
\end{equation}

\noindent holds for some $\bm{U}_{d}\in\mathcal{U}_{\mathrm{ILC}}(\bm{Y}_{d})$ given by (\ref{eq25}) if and only if
\begin{equation}\label{e11}
\rho\left(I-\bm{F}_{1}^{\tp}\bm{G}\bm{K}\bm{H}_{1}\right)=0
\end{equation}

\noindent where $\nu_{1}$ is the minimal integer such that
\begin{equation}\label{e12}
\left(I-\bm{F}_{1}^{\tp}\bm{G}\bm{K}\bm{H}_{1}\right)^{\nu_{1}-1}\neq0~\hbox{and}~
\left(I-\bm{F}_{1}^{\tp}\bm{G}\bm{K}\bm{H}_{1}\right)^{\nu_{1}}=0.
\end{equation}

\noindent For $\bm{Y}_{d}\notin\mathcal{Y}_{T}$, there exists some $\nu_{2}\in\mathbb{Z}_{m}$ such that
\begin{equation}\label{e014}
\bm{U}_{k}
\left\{\aligned
&\neq\overline{\bm{U}}_{d}, &\forall k&\leq\nu_{2}-1\\
&=\overline{\bm{U}}_{d}, &\forall k&\geq\nu_{2}
\endaligned
\right.,
~~
\bm{Y}_{k}
\left\{\aligned
&\neq\bm{H}_1\bm{F}_1^{\tp}\bm{Y}_{d}, &\forall k&\leq\nu_{2}-1\\
&=\bm{H}_1\bm{F}_1^{\tp}\bm{Y}_{d}, &\forall k&\geq\nu_{2}
\endaligned
\right.
\end{equation}

\noindent holds for some $\overline{\bm{U}}_{d}\in\overline{\mathcal{U}}_{\mathrm{ILC}}(\bm{Y}_{d})$ given by (\ref{e3}) if and only if
\begin{equation}\label{e13}
\rho\left(I-\bm{F}_{1}^{\tp}\bm{G}\widehat{\bm{K}}\right)=0.
\end{equation}

\noindent Further, $\bm{K}=\bm{\widehat{K}}\bm{F}^{\tp}_{1}$ holds for (\ref{eq19}), and $\nu_{2}$ is the minimal integer such that
\begin{equation}\label{e14}
\left(I-\bm{F}_{1}^{\tp}\bm{G}\widehat{\bm{K}}\right)^{\nu_{2}-1}\neq0~\hbox{and}~
\left(I-\bm{F}_{1}^{\tp}\bm{G}\widehat{\bm{K}}\right)^{\nu_{2}}=0.
\end{equation}
\end{thm}

\begin{proof}
See Appendix C.
\end{proof}

\begin{rem}\label{rem3}
With Theorem \ref{t3}, a finite iteration convergence result of ILC is obtained, regardless of trackable or untrackable desired references. This is realized by appropriate selections of the ILC gain matrix. In addition, the number $\nu_{1}$ (respectively, $\nu_{2}$) of iteration steps for ILC convergence is exactly the degree of the minimal polynomial for $I-\bm{F}_{1}^{\tp}\bm{G}\bm{K}\bm{H}_{1}$ (respectively, $I-\bm{F}_{1}^{\tp}\bm{G}\widehat{\bm{K}}$) versus the trackable (respectively, untrackable) desired reference. It actually incorporates the idea of deadbeat control into the improvement of the convergence rate for ILC, which is thus applicable for improving the convergence rate of iterative methods for solving LAEs from the control design perspective. 
%
\end{rem}

\section{Problem-Solving of LAEs Via ILC Methods}\label{sec6}

In this section, we consider how to leverage the ILC method to solve LAEs. Before proceeding further, we present a lemma for the gain matrix design of ILC.

\begin{lem}\label{lem09}
For any $\bm{G}\in\mathbb{R}^{p\times q}$, there exists some $\bm{K}\in\mathbb{R}^{q\times p}$ to satisfy (\ref{eq20}) (respectively, (\ref{e11})) if and only if there exists some $\widehat{\bm{K}}\in\mathbb{R}^{q\times m}$ to satisfy (\ref{eq21}) (respectively, (\ref{e13})).
\end{lem}

\begin{proof}
{\it Sufficiency:} Due to $\bm{F}_{1}^{\tp}\bm{H}_{1}=I$, we take $\bm{K}=\widehat{\bm{K}}\bm{F}_{1}^{\tp}$ and can arrive at $\bm{F}_{1}^{\tp}\bm{G}\bm{K}\bm{H}_{1}=\bm{F}_{1}^{\tp}\bm{G}\widehat{\bm{K}}\left(\bm{F}_{1}^{\tp}\bm{H}_{1}\right)
=\bm{F}_{1}^{\tp}\bm{G}\widehat{\bm{K}}$. Thus, we can obtain (\ref{eq20}) (respectively, (\ref{e11})) from (\ref{eq21}) (respectively, (\ref{e13})) by taking $\bm{K}=\widehat{\bm{K}}\bm{F}_{1}^{\tp}$.

{\it Necessity:} Let $\widehat{\bm{K}}=\bm{K}\bm{H}_{1}$, and then it is immediate to deduce (\ref{eq21}) (respectively, (\ref{e13})) from (\ref{eq20}) (respectively, (\ref{e11})).
\end{proof}

By Lemma \ref{lem09}, we always design the updating law (\ref{eq19}) with $\bm{K}=\widehat{\bm{K}}\bm{F}_{1}^{\tp}$ under the conditions (\ref{eq21}) and (\ref{e13}) without loss of generality. Since $\bm{F}_{1}^{\tp}\bm{G}$ is of the full-row rank based on (\ref{e019}), we can obtain the candidate selections of $\widehat{\bm{K}}$, respectively, fulfilling (\ref{eq21}) as
\begin{equation}\label{e017}
\widehat{\bm{K}}=\bm{G}^{\tp}\bm{F}_{1}\left(\bm{F}_{1}^{\tp}\bm{G}\bm{G}^{\tp}\bm{F}_{1}\right)^{-1}\left(I-\overline{\bm{K}}\right)
\end{equation}

\noindent and (\ref{e13}) as
\begin{equation}\label{e018}
\widehat{\bm{K}}=\bm{G}^{\tp}\bm{F}_{1}\left(\bm{F}_{1}^{\tp}\bm{G}\bm{G}^{\tp}\bm{F}_{1}\right)^{-1}\left(I-\widetilde{\bm{K}}\right)
\end{equation}

\noindent where $\overline{\bm{K}}\in\mathbb{R}^{m\times m}$ is any stable matrix (i.e., its eigenvalues are within the unit circle), and $\widetilde{\bm{K}}\in\mathbb{R}^{m\times m}$ is any nilpotent matrix (i.e., $\widetilde{\bm{K}}^{n}=0$ holds for some positive integer $n$ \cite{lt:85}).

With Lemma \ref{lem09}, a fundamental solvability result is proposed for LAEs in the following theorem.

\begin{thm}\label{thm09}
Given $\bm{Y}_{d}\in\mathbb{R}^{p}$, the LAE (\ref{eq1}) is solvable if and only if $\bm{Y}_{d}$ is trackable for the system (\ref{eq2}), where the application of the updating law (\ref{eq19}) with $\bm{K}=\widehat{\bm{K}}\bm{F}_{1}^{\tp}$ based on the condition (\ref{eq21}) determines the solution to the LAE (\ref{eq1}) in the convergence form of $\lim_{k\to\infty}\bm{U}_{k}$, which is collected with an analytical form in the convex set (\ref{eq25}). Further, if $\widehat{\bm{K}}$ is designed according to (\ref{e13}), then the solution to the LAE (\ref{eq1}) can be determined within finite iterations.
\end{thm}

\begin{proof}
With Definition \ref{defi2}, this theorem can be developed based on Theorems \ref{thm07}, \ref{thm08}, and \ref{t3}.
\end{proof}

From Lemma \ref{lem02}, we can resort to $\bm{H}_{1}\bm{F}_{1}^{\tp}\bm{Y}_d=\bm{Y}_d$ to establish the solvability of the LAE (\ref{eq1}) in Theorem \ref{thm09} for any $\bm{Y}_{d}\in\mathbb{R}^{p}$. We also note from (\ref{eq25}) and (\ref{eq26}) that the solutions to the LAE (\ref{eq1}) depend on the selection of the initial input $\bm{U}_{0}$ of the system (\ref{eq2}), and thus we may determine all solutions through choosing all candidates of $\bm{U}_{0}$ in $\mathbb{R}^{q}$.
%

Next, we explore Theorem \ref{thm09} to present an implementation algorithm of ILC to determine solutions for solvable LAEs.

{\bf Algorithm 1: Solving the LAE (\ref{eq1}) for any $\bm{Y}_{d}\in\mathcal{Y}_{T}$}
\begin{enumerate}
\item[1)]
Set a tolerance $\varepsilon>0$, calculate $\bm{K}=\widehat{\bm{K}}\bm{F}_{1}^{\tp}$ by determining $\bm{F}_{1}$ and designing $\widehat{\bm{K}}$ according to (\ref{eq21}) or (\ref{e13}), and select the initial input $\bm{U}_{0}$.

\item[2)]
Let $k=0$, and go to the step 3) to start the iteration.

\item[3)]
Apply $\bm{U}_{k}$ to the system (\ref{eq2}) to measure $\bm{Y}_{k}$.

\item[4)]
Calculate $\bm{E}_{k}=\bm{Y}_{d}-\bm{Y}_{k}$. If $\|\bm{E}_{k}\|<\epsilon$, then go to the step 7) by returning the solution ``$\bm{U}_{d}=\bm{U}_{k}$;'' and otherwise, go to the step 5).

\item[5)]
Apply the updating law (\ref{eq19}) to calculate $\bm{U}_{k+1}$.

\item[6)]
Let $k=k+1$, and return to the step 3).

\item[7)]
Stop the iteration.
\end{enumerate}

To determine $\bm{F}_{1}$ in the Algorithm 1, we only need to choose $\bm{H}_{1}$ as a full-column rank matrix that satisfies $\sn\bm{H}_{1}=\sn\bm{G}$, and can directly take $\bm{F}_{1}=\bm{H}_{1}\left(\bm{H}_{1}^{\tp}\bm{H}_{1}\right)^{-1}$. This is thanks to the property that the selection of $\bm{H}_{2}$, together with that of $\bm{F}_{2}$, does not have influence on both design condition and convergence result of ILC though it plays a significant role in implementing the convergence analysis of ILC.

When the LAE (\ref{eq1}) is unsolvable, the system (\ref{eq2}) is subject to the untrackable desired references. However, the ILC updating law \eqref{eq19} is still applicable, and helps calculate the least squares solutions to the LAE \eqref{eq1}, as shown in the following theorem.

\begin{thm}\label{t2}
For the LAE (\ref{eq1}) with any $\bm{Y}_{d}\notin\mathcal{Y}_{T}$, the system (\ref{eq2}) with the application of the updating law (\ref{eq19}) for $\bm{K}=\widehat{\bm{K}}\bm{F}_{1}^{\tp}$ under the condition (\ref{eq21}) determines its least squares solution in the convergence form of $\lim_{k\to\infty}\bm{U}_{k}$, which is collected with an analytical form in the convex set (\ref{e3}) if and only if $\bm{H}_{1}^{\tp}\bm{H}_{2}=0$ holds for the selection of $\bm{H}$. Moreover, if $\widehat{\bm{K}}$ is designed based on (\ref{e13}), then the least squares solution to the LAE (\ref{eq1}) can be determined within finite iterations.
\end{thm}

\begin{proof}
A consequence of applying Theorems \ref{t1} and \ref{t3}.
\end{proof}

With Theorem \ref{t2}, we show an implementation algorithm of ILC to determine least squares solutions for unsolvable LAEs.

{\bf Algorithm 2: Solving the LAE (\ref{eq1}) for any $\bm{Y}_{d}\notin\mathcal{Y}_{T}$}
\begin{enumerate}
\item[1)]
Select a full-column rank matrix $\bm{H}_{1}$ such that $\sn\bm{H}_{1}=\sn\bm{G}$, and calculate $\bm{F}_{1}=\bm{H}_{1}\left(\bm{H}_{1}^{\tp}\bm{H}_{1}\right)^{-1}$.

\item[2)]
Set a tolerance $\varepsilon>0$, calculate $\bm{K}=\widehat{\bm{K}}\bm{F}_{1}^{\tp}$ by determining $\widehat{\bm{K}}$ based on (\ref{eq21}) or (\ref{e13}), and select an initial input $\bm{U}_{0}$.

\item[3)]
Let $k=0$, and go to the step 4) to start the iteration.

\item[4)]
Apply $\bm{U}_{k}$ to the system (\ref{eq2}) to measure $\bm{Y}_{k}$.

\item[5)]
Calculate $\overline{\bm{E}}_{k}=\bm{H}_{1}\bm{F}_{1}^{\tp}\bm{Y}_{d}-\bm{Y}_{k}$. If $\left\|\overline{\bm{E}}_{k}\right\|<\epsilon$, then go to the step 8) by returning the least squares solution ``$\overline{\bm{U}}_{d}=\bm{U}_{k}$;'' and otherwise, go to the step 6).

\item[6)]
Apply the updating law (\ref{eq19}) to calculate $\bm{U}_{k+1}$.

\item[7)]
Let $k=k+1$, and return to the step 4).

\item[8)]
Stop the iteration.
\end{enumerate}

For the Algorithm 2, the use of $\bm{F}_{1}=\bm{H}_{1}\left(\bm{H}_{1}^{\tp}\bm{H}_{1}\right)^{-1}$ ensures $\bm{H}_{1}^{\tp}\bm{H}_{2}=0$ since the selection of $\bm{H}_{2}$ needs to satisfy $\bm{F}=\bm{H}^{-1}$. This renders the implementation of the Algorithm 2 capable of determining the least squares solutions for LAEs. To determine whether the Algorithm 1 or 2 should be used for any LAE (\ref{eq1}), it depends on whether $\bm{Y}_{d}\in\mathcal{Y}_{T}$ or $\bm{Y}_{d}\notin\mathcal{Y}_{T}$, which can instead be validated by whether the criterion $\bm{H}_{1}\bm{F}_{1}^{\tp}\bm{Y}_d=\bm{Y}_d$ is satisfied or not. 
In particular, for the design of $\widehat{\bm{K}}$, we can directly use (\ref{e017}) and (\ref{e018}) that render the Algorithms 1 and 2 exponentially convergent and convergent within finite iterations, respectively.

To explain how to solve LAEs with our ILC algorithms, we provide illustrations in the following example.

{\it Example 1:} Consider the LAE (\ref{eq1}) with $\bm{G}$ given by
\[
\bm{G}=\left[\begin{matrix}
1&0&0&2&2\\
0&3&4&0&0\\
1&-3&-4&2&2
\end{matrix}\right]^{\tp}.
\]

\noindent Clearly, $\rank(\bm{G})=2$ holds. To select a full-column rank matrix $\bm{H}_{1}$ such that $\sn\bm{H}_{1}=\sn\bm{G}$, we adopt
\[
\bm{H}_1=\left[\begin{matrix}
1&0&0&2&2\\
0&3&4&0&0
\end{matrix}\right]^{\tp}.
\]

\noindent As a consequence, to apply the Algorithms 1 and 2, we directly choose $\bm{F}_{1}$ as
\[
\bm{F}_{1}=\bm{H}_{1}\left(\bm{H}_{1}^{\tp}\bm{H}_{1}\right)^{-1}=\left[\begin{matrix}
1/9&0&0&2/9&2/9\\
0&3/25&4/25&0&0
\end{matrix}\right]^{\tp}
\]

\noindent and adopt $\bm{K}=\widehat{\bm{K}}\bm{F}_{1}^{\tp}$, where we particularly design $\widehat{\bm{K}}$ according to (\ref{e018}). Namely, by specifying a nilpotent matrix as
\[\widetilde{\bm{K}}=\begin{bmatrix}0&1\\0&0\end{bmatrix}\]

\noindent we can calculate $\widehat{\bm{K}}$ as
\[
\widehat{\bm{K}}=\bm{G}^{\tp}\bm{F}_{1}\left(\bm{F}_{1}^{\tp}\bm{G}\bm{G}^{\tp}\bm{F}_{1}\right)^{-1}\left(I-\widetilde{\bm{K}}\right)
=\left[\begin{matrix}
2/3&1/3&1/3\\
-1/3&1/3&-2/3
\end{matrix}\right]^{\tp}.
\]

\noindent Then we can leverage the Algorithms 1 and 2 to solve the LAE (\ref{eq1}) for any desired reference $\bm{Y}_{d}\in\mathbb{R}^{5}$. Due to $I-\bm{F}_{1}^{\tp}\bm{G}\widehat{\bm{K}}=\widetilde{\bm{K}}$ and $\widetilde{\bm{K}}^{2}=0$, we can derive the calculation result after only two iterations of the implementation of the Algorithms 1 and 2.

To proceed, we without any loss of generality use the initial input as $\bm{U}_0=[1,0,0]^{\tp}$, and consider two different cases of the desired references as
\begin{enumerate}
\item[c1)]
$\bm{Y}_{d}=\left[1, 3, 4, 2, 2\right]^{\tp}$;

\item[c2)]
$\bm{Y}_{d}=\left[1, 2, 1, 1, 2\right]^{\tp}$.
\end{enumerate}

\noindent We can easily verify that for the case c1), $\bm{H}_{1}\bm{F}_{1}^{\tp}\bm{Y}_{d}=\bm{Y}_{d}$ holds, whereas for the case c2), it does not. This implies that the LAE \eqref{eq1} is solvable in the case c1), but not in the case c2).

For the case c1), we apply the Algorithm 1, and can obtain the solution to the LAE \eqref{eq1} after two iterations as
\[\bm{U}_{d}=\left[4/3, 2/3, -1/3\right]^{\tp}.
\]

\noindent In fact, we can determine the solutions to the LAE (\ref{eq1}) for any trackable desired reference $\bm{Y}_{d}\in\mathcal{Y}_{T}$ in a general form of
\[
\bm{U}_{d}=\begin{bmatrix}
1/3&-1/3&-1/3\\
-1/3&1/3&1/3\\
-1/3&1/3&1/3
\end{bmatrix}\bm{U}_0+\begin{bmatrix}
1\\
1\\
0
\end{bmatrix},\quad\forall\bm{U}_{0}\in\mathbb{R}^{q}.
\]

For the case c2), the application of the Algorithm 2 results in the least squares solution to the LAE (\ref{eq1}) after two iterations, which is given by
\[\overline{\bm{U}}_{d}=\left[133/135, 26/135, -28/135\right]^{\tp}
\]

\noindent and as a result, the least squares error norm is $\left\|\bm{Y}_{d}-\bm{G}\overline{\bm{U}}_{d}\right\|_{2}=\sqrt{14}/{3}$. In addition, the least squares solutions to the LAE (\ref{eq1}) for any untrackable desired reference $\bm{Y}_{d}\notin\mathcal{Y}_{T}$ have a general form of
\[
\overline{\bm{U}}_{d}=\begin{bmatrix}
1/3&-1/3&-1/3\\
-1/3&1/3&1/3\\
-1/3&1/3&1/3
\end{bmatrix}\bm{U}_0+\begin{bmatrix}
88/135\\
71/135\\
17/135
\end{bmatrix},\quad\forall\bm{U}_{0}\in\mathbb{R}^{q}.
\]

\section{Applications to Traditional 2-D ILC Systems}\label{sec7}

Traditionally, each ILC system involves dynamics along two independent axes. In addition to an infinite iteration axis given by $k\in\mathbb{Z}_{+}$, it generally has a fixed time axis denoted by $t\in\mathbb{Z}_{N}$. We consider the ILC problem for a system described by
\begin{equation}\label{eq42}
\left\{\aligned
x_{k}(t+1)
&=Ax_{k}(t)+Bu_{k}(t)\\
y_{k}(t)
&=Cx_{k}(t)
\endaligned\right.,\quad\forall t\in\mathbb{Z}_{N},\forall k\in\mathbb{Z}_{+}
\end{equation}

\noindent where $x_{k}(t)\in\mathbb{R}^{n_{s}}$, $u_{k}(t)\in\mathbb{R}^{n_{i}}$ and $y_{k}(t)\in\mathbb{R}^{n_{o}}$ denote the state, input and output, respectively; and $A$, $B$ and $C$ are three system matrices of appropriate dimensions. The implementation of the system (\ref{eq42}) requires two different initial conditions, which are given by $x_{k}(0)=0$, $\forall k\in\mathbb{Z}_{+}$ and $u_{0}(t)$, $\forall t\in\mathbb{Z}_{N-1}$. In addition, let us denote $C=\left[C_{1}^{\tp},C_{2}^{\tp},\cdots,C_{n_{o}}^{\tp}\right]^{\tp}$, where $C_{i}\in\mathbb{R}^{1\times n_{s}}$, $\forall i=1$, $2$, $\cdots$, $n_{o}$.

Let each of the $n_{o}$ input-output channels for the system (\ref{eq42}) have a relative degree $r$ ($r\geq1$) without any loss of generality. Namely, it holds
\begin{equation}\label{eq43}
\aligned
C_{i}A^{j}B&=0,j=0,1,\cdots,r-2\\
C_{i}A^{r-1}B&\neq0
\endaligned,\quad\forall i=1,2,\cdots,n_{o}.
\end{equation}

\noindent Based on some standard algebraic manipulations, the condition (\ref{eq43}) can be generalized to describe different relative degrees $r_{i}$, $\forall i=1$, $2$, $\cdots$, $n_{o}$ for the system (\ref{eq42}) (see, e.g., \cite{sw:02,s:05}). This generalization, however, does not influence our investigations, which will not be considered. Due to (\ref{eq43}), we denote the series of the Markov parameter matrices for the system (\ref{eq42}) as
\[
G_{i}=CA^{i+r-1}B,\quad\forall i\in\mathbb{Z}_{+}.
\]

\noindent It is worth noting that $G_{0}=CA^{r-1}B$ is the first nonzero Markov parameter matrix for the system (\ref{eq42}), and also that the system relative degree $r$ is exactly the delay of time steps in the output $y_{k}(t)$ when the input $u_{k}(t)$ takes effect \cite{sw:03}.

For the system (\ref{eq42}), the evolution from iteration to iteration is to enable its output to track the desired reference for all time steps of interest over a fixed time interval. If we use $y_{d}(t)\in\mathbb{R}^{n_{o}}$ to represent any desired reference, then in view of the relative degree $r$ for the system (\ref{eq42}), the tracking objective is to design appropriate updating laws of $u_{k}(t)$, $\forall t\in\mathbb{Z}_{N-1}$ such that
\begin{equation}\label{eq44}
\lim_{k\to\infty}y_{k}(t)=y_{d}(t),\quad\forall t=r,r+1,\cdots,r+N-1.
\end{equation}

To cope with this tracking problem of ILC, we are motivated by Definition \ref{defi2} to give the trackability and realizability notions of the desired references for the system (\ref{eq42}).

\begin{defi}\label{defi3}
For the system (\ref{eq42}), a desired reference $y_{d}(t)$ is said to be trackable (respectively, realizable) if there exists some (respectively, a unique) desired input $u_{d}(t)$, together with the zero initial state condition $x_{d}(0)=0$, such that
\begin{equation}\label{eq45}
\left\{\aligned
x_{d}(t+1)
&=Ax_{d}(t)+Bu_{d}(t)\\
y_{d}(t)
&=Cx_{d}(t)
\endaligned\right.,\quad\forall t\in\mathbb{Z}_{N}.
\end{equation}
\end{defi}

In Definition \ref{defi3}, we may make extensions by considering any initial state condition $x_{d}(0)\in\mathbb{R}^{n_{s}}$. Since $x_{d}(0)$ is closely tied to $y_{d}(t)$, we employ $y_{d}(t)-CA^{t}x_{d}(0)$ instead of directly applying $y_{d}(t)$ to bridge the connection with $u_{d}(t)$. By this replacement, we can obtain the same trackability and realizability properties, thanks to which we directly adopt $x_{d}(0)=0$ to keep consistent with the establishment of Definition \ref{defi2}.

\subsection{ILC Reformulation Based on Super-Vectors}

Let $p=Nn_{o}$ and $q=Nn_{i}$, and define $\bm{Y}_{k}$, $\bm{Y}_{d}$, $\bm{U}_{k}$ and $\bm{U}_{d}$ as the super-vectors that result from $y_{k}(t)$, $y_{d}(t)$, $u_{k}(t)$ and $u_{d}(t)$, respectively, based on the lifting technique in the form of
\begin{equation}\label{eq46}
\aligned
\bm{Y}_{k}
&=\left[y_{k}^{\tp}(r),y_{k}^{\tp}(r+1),\cdots,y_{k}^{\tp}(N+r-1)\right]^{\tp},~\forall k\in\mathbb{Z}_{+}\\
\bm{Y}_{d}
&=\left[y_{d}^{\tp}(r),y_{d}^{\tp}(r+1),\cdots,y_{d}^{\tp}(N+r-1)\right]^{\tp}\\
\bm{U}_{k}
&=\left[u_{k}^{\tp}(0),u_{k}^{\tp}(1),\cdots,u_{k}^{\tp}(N-1)\right]^{\tp},~\forall k\in\mathbb{Z}_{+}\\
\bm{U}_{d}
&=\left[u_{d}^{\tp}(0),u_{d}^{\tp}(1),\cdots,u_{d}^{\tp}(N-1)\right]^{\tp}.
\endaligned
\end{equation}

\noindent We can reformulate (\ref{eq42}) and (\ref{eq45}) into (\ref{eq2}) and (\ref{eq1}), respectively, by defining $\bm{G}$ in the structured form of a block Toeplitz matrix:
\begin{equation}\label{eq47}
\bm{G}
=\begin{bmatrix}
G_{0}&0&\cdots&0\\
G_{1}&G_{0}&\ddots&\vdots\\
\vdots&\vdots&\ddots&0\\
G_{N-1}&G_{N-2}&\cdots&G_{0}
\end{bmatrix}.
\end{equation}
%

\noindent Similarly to (\ref{eq46}), $\bm{E}_{k}$ can be defined as
\begin{equation*}\label{}
\bm{E}_{k}
=\left[e_{k}^{\tp}(r),e_{k}^{\tp}(r+1),\cdots,e_{k}^{\tp}(N+r-1)\right]^{\tp},\quad\forall k\in\mathbb{Z}_{+}
\end{equation*}

\noindent in which $e_{k}(t)=y_{d}(t)-y_{k}(t)$ denotes the output tracking error of the system (\ref{eq42}). 

Since the matrix $\bm{G}$ has a special structure of (\ref{eq47}), we benefit from this property to get some rank conditions on $\bm{G}$, especially with the aid of the rank conditions on $G_{0}$.

\begin{lem}\label{lem04}
For the matrix $\bm{G}$ defined by (\ref{eq47}), it holds:
\begin{enumerate}
\item
$N\rank\left(G_{0}\right)\leq\rank\left(\bm{G}\right)\leq N\min\left\{n_{o},n_{i}\right\}$;
%

\item
$\rank\left(\bm{G}\right)=N\min\left\{n_{o},n_{i}\right\}\Leftrightarrow\rank\left(G_{0}\right)=\min\left\{n_{o},n_{i}\right\}$.
\end{enumerate}
\end{lem}

\begin{proof}
A direct consequence of the block Toeplitz matrix structure of $\bm{G}$ in (\ref{eq47}).
\end{proof}

By Lemma \ref{lem04}, we generally have $\rank\left(\bm{G}\right)\geq N\rank\left(G_{0}\right)$, and $\rank\left(\bm{G}\right)=N\rank\left(G_{0}\right)$ emerges when $\rank\left(G_{0}\right)=\min\left\{n_{o},n_{i}\right\}$. We also find that $\bm{G}$ is of full row (respectively, column) rank if and only if $G_{0}$ is of full row (respectively, column) rank. These properties can help achieve the conditions of $\bm{G}$ in establishing the trackability and realizability properties for ILC systems. Of particular note is this benefit for the system (\ref{eq2}) that is induced from the system (\ref{eq42}).

From Definitions \ref{defi2} and \ref{defi3} and with (\ref{eq46}), we can establish the trackability and realizability properties of $y_{d}(t)$ for the system (\ref{eq42}) by exploring those of $\bm{Y}_{d}$ for the system (\ref{eq2}). Consequently, we also denote $\mathcal{Y}_{T}$ and $\mathcal{Y}_{R}$ as the trackability and realizability subspaces of the system (\ref{eq42}), respectively. Besides, a trackable (respectively, realizable) desired reference $y_{d}(t)$ is denoted by $y_{d}(t)\in\mathcal{Y}_{T}$ (respectively, $y_{d}(t)\in\mathcal{Y}_{R}$).

With $\bm{E}_{k}$, we can still arrive at the system (\ref{eq4}) to deal with the tracking problem (\ref{eq44}) for the system (\ref{eq42}). Based on Definition \ref{defi1}, it follows that we can achieve the tracking objective (\ref{eq44}) for the system (\ref{eq42}) by transforming it into the $k$-stability problem of the system (\ref{eq4}).

\subsection{Trackability-Based ILC Analysis and Design}

For the trackability of the system (\ref{eq42}), we present a helpful lemma to provide it with some basic properties.

\begin{lem}\label{lem05}
For the system (\ref{eq42}), $\mathcal{Y}_{T}=\sn\bm{G}$ always holds, from which two properties follow as:
\begin{enumerate}
\item
$N\rank\left(G_{0}\right)\leq\dim\left(\mathcal{Y}_{T}\right)\leq N\min\left\{n_{o},n_{i}\right\}$;

\item
$\mathcal{Y}_{T}=\mathbb{R}^{Nn_{o}}$ if and only if $\rank\left(G_{0}\right)=n_{o}$.
\end{enumerate}

\noindent In particular, for any desired reference $y_{d}(t)$, $y_{d}(t)\in\mathcal{Y}_{T}$ if and only if $\bm{Y}_{d}\in\sn\bm{G}$.
\end{lem}

\begin{proof}
With the applications of Theorem \ref{thm01} and Lemmas \ref{lem02} and \ref{lem04}, this lemma can be proved by noting the lower triangular block form of (\ref{eq47}) for the matrix $\bm{G}$.
%
%
%
\end{proof}

%

Analogously to Lemma \ref{lem05}, the following lemma shows basic realizability properties of the system (\ref{eq42}).

\begin{lem}\label{lem06}
For the system (\ref{eq42}), $\mathcal{Y}_{R}=\sn\bm{G}$ follows if and only if $\rank\left(G_{0}\right)=n_{i}$ holds; and $\mathcal{Y}_{R}=\{0\}$ emerges, otherwise. In particular, for any desired reference $y_{d}(t)$, $y_{d}(t)\in\mathcal{Y}_{R}$ if and only if $y_{d}(t)\in\mathcal{Y}_{T}$ and $\rank\left(G_{0}\right)=n_{i}$.
\end{lem}

\begin{proof}
A consequence of Lemma \ref{lem03} and Theorem \ref{thm01} thanks to $\bm{G}$ given by (\ref{eq47}).
\end{proof}

With Lemmas \ref{lem05} and \ref{lem06}, we can notice by comparison that the proposed trackability property may play a more fundamentally important role than the usually employed realizability property in addressing the design and analysis problems of ILC. To proceed, we explore the controllability properties of the tracking error system (\ref{eq4}) with $\bm{G}$ given by (\ref{eq47}).

\begin{lem}\label{lem07}
The system (\ref{eq4}) is controllable (or stabilizable) if and only if $\rank\left(G_{0}\right)=n_{o}$, and admits a standard form for the controllability decomposition in (\ref{eq15}), otherwise. Furthermore, $\mathcal{C}_{C}=\mathcal{Y}_{T}=\sn\bm{G}$ always holds.
\end{lem}

\begin{proof}
With the specific structure of $\bm{G}$ in (\ref{eq47}), this lemma can follow directly from Theorems \ref{thm03}, \ref{thm04} and \ref{thm05}.
\end{proof}

Based on the development of Lemmas \ref{lem05} and \ref{lem07}, we establish a basic tracking result of ILC for the system (\ref{eq42}).

\begin{thm}\label{thm02}
For the system (\ref{eq42}) with any desired reference $y_{d}(t)$, the tracking objective (\ref{eq44}) can be achieved under some updating law of ILC, with gain matrices $K_{ij}\in\mathbb{R}^{n_{i}\times n_{o}}$, $\forall i$, $j=1$, $2$, $\cdots$, $N$, given by
\begin{equation}\label{eq51}
u_{k+1}(t)=u_{k}(t)+\sum_{i=0}^{N-1}K_{t+1,i+1}e_{k}(i+r),\quad\forall t\in\mathbb{Z}_{N-1},\forall k\in\mathbb{Z}_{+}
\end{equation}

\noindent if and only if $y_{d}(t)\in\mathcal{Y}_{T}$; and otherwise, the tracking objective (\ref{eq44}) can no longer be achieved for the system (\ref{eq42}), regardless of applying any input $u_{k}(t)$, $\forall t\in\mathbb{Z}_{N-1}$, $\forall k\in\mathbb{Z}_{+}$. In particular, the design condition of the gain matrices in (\ref{eq51}) for any $y_{d}(t)\in\mathcal{Y}_{T}$ is such that $\bm{K}=\left[K_{ij}\right]\in\mathbb{R}^{Nn_{i}\times Nn_{o}}$ fulfills the spectral radius condition (\ref{eq20}).
\end{thm}

\begin{proof}
By noting the equivalence between (\ref{eq19}) and (\ref{eq51}), we can develop this theorem with Lemmas \ref{lem05} and \ref{lem07} and based on Theorem \ref{thm07} as well as its proof.
\end{proof}

\begin{rem}\label{rem10}
From Theorem \ref{thm07}, we know that the ILC results of Theorem \ref{thm02} benefit from the trackability result of Lemma \ref{lem05}, the controllability result of Lemma \ref{lem07}, and the $k$-state feedback controller design in the Kalman state-space framework. It hints that our trackability-based approach can provide a feasible way to narrow the gap between ILC and the popular state feedback-based control methods. As a result, more powerful design and analysis tools may be developed for ILC, such as those gained based on the system stability (see also Remark \ref{rem5}).
\end{rem}

Next, we use an example to demonstrate the effectiveness of ILC even in the absence of full (row or column) rank condition of the controlled system.

{\it Example 2:} Consider the system (\ref{eq42}) with $A$, $B$, and $C$ as
\[
A=\begin{bmatrix}
1&0&0\\
0&1&0\\
0&0&1\end{bmatrix},\quad
B=\begin{bmatrix}1&-1\\2&-2\\0&0\end{bmatrix},\quad
C=\begin{bmatrix}1&0&1\\0&1&-1
\end{bmatrix}.
\]

\noindent Clearly, $G_{0}=CB\in\mathbb{R}^{2\times2}$ is such that $\rank\left(G_{0}\right)=1$. That is, $G_{0}$ is neither of
full-row rank nor of full-column rank. It generally renders the existing results for ILC not applicable. By contrast, we reveal in Theorem \ref{thm02} that we can still perform the tracking tasks of ILC in the presence of trackable desired references.

As an example, we consider the desired reference given by
\[y_{d}(t)=\left[\sin(0.06t), 2\sin(0.06t)\right]^{\tp},\quad\forall t\in\mathbb{Z}_{100}.
\]

\noindent It is not difficult to know $y_{d}(t)\in\mathcal{Y}_{T}$. To apply the updating law (\ref{eq51}), we directly take $\bm{K}=I\otimes K_{0}\in\mathbb{R}^{200\times200}$, and to determine $K_{0}\in\mathbb{R}^{2\times2}$, we select $\bm{H}_{1}=\bm{G}\bm{Z}$ with $\bm{Z}=I\otimes[1,0]^{\tp}\in\mathbb{R}^{200\times100}$ and $\bm{H}_{2}=I\otimes[0,1]^{\tp}\in\mathbb{R}^{200\times100}$. In this case, if we choose $K_{0}=diag\{1.5,0.73\}$, then we can obtain $\rho\left(I-\bm{F}_{1}^{\tp}\bm{G}\bm{K}\bm{H}_{1}\right)=0.96$, namely, the spectral radius condition (\ref{eq20}) holds. In Fig. \ref{fig1}, we depict the perfect tracking performance for ILC under the zero initial input. This illustration demonstrates the effectiveness of our trackability-based ILC results of Theorem \ref{thm02}.

\begin{figure}
\centering
\includegraphics[width=3in]{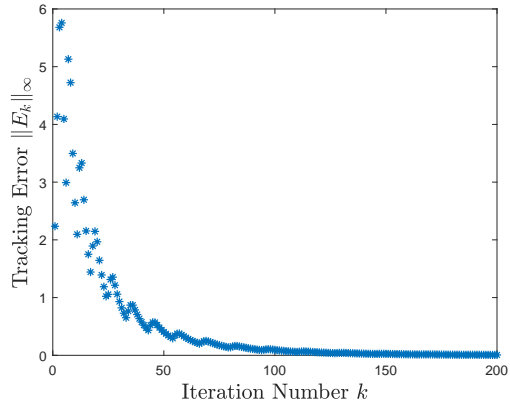}
\caption{The evolution of the tracking error for ILC versus iteration number.}
\label{fig1}
\end{figure}

\section{Conclusions}\label{sec8}

In this paper, we have discussed and answered the questions about whether and how we can achieve the objective of solving LAEs from the perspective of systems and control, concerning which we have developed an iterative method integrated with a learning control mechanism. By introducing a new trackability problem and connecting it to typical tracking problems of ILC, we have revealed an equivalent relation between the solvability of LAEs and the controllability of discrete control systems. In view of this equivalence relation, we have leveraged the classic state feedback-based methods to solve LAEs, which is realized by equivalently accomplishing the perfect tracking tasks of the resulting ILC systems. It has been shown that the solutions for any solvable LAE are linearly dependent upon the initial inputs and thus all its solutions can be derived by different selections of the initial inputs. This property also applies to obtaining all the least squares solutions for any unsolvable LAE. Moreover, we have incorporated the design idea of deadbeat control into ILC such that the solving of LAEs can be realized within finite iteration steps. In addition, we have generalized the trackability property to conventional 2-D ILC systems, and have developed a state feedback-based method for the analysis and synthesis of ILC, instead of resorting to the generally employed contraction mapping-based methods for ILC. This brings new insights into ILC and narrows the gap between it and the popular feedback-based control methods. Our analysis and synthesis results may also provide a way to promote the interaction between control and mathematics.

\section*{Acknowledgement}\label{sec}

We would like to express our thanks to Dr. Jingyao Zhang, Beihang University (BUAA), for his helpful discussion on the development of Theorem \ref{t3}.

\appendices

\section{Proof of Theorem \ref{thm07}}

\begin{proof}
Two steps are included to develop this proof.

{\it Step (i): Let us prove that for any $\bm{Y}_{d}\in\mathcal{Y}_{T}$, $\bm{U}_{\infty}$ exists and satisfies $\bm{U}_{\infty}\in\mathcal{U}_{\mathrm{ILC}}(\bm{Y}_{d})$.}

From (\ref{eq2}) and (\ref{eq19}), we can arrive at
\begin{equation}\label{eq28}
\bm{U}_{k+1}=\left(I-\bm{K}\bm{G}\right)\bm{U}_{k}+\bm{K}\bm{Y}_{d},\quad\forall k\in\mathbb{Z}_{+}.
\end{equation}

\noindent By combining two results of Lemma \ref{lem01} with (\ref{eq28}), we can derive
\begin{equation}\label{eq29}
\aligned
\bm{U}_{k+1}
&=\left[I-\bm{K}\left(\bm{H}_{1}\bm{F}_{1}^{\tp}+\bm{H}_{2}\bm{F}_{2}^{\tp}\right)\bm{G}\right]\bm{U}_{k}+\bm{K}\bm{Y}_{d}\\
&=\left(I-\bm{K}\bm{H}_{1}\bm{F}_{1}^{\tp}\bm{G}\right)\bm{U}_{k}+\bm{K}\bm{Y}_{d},\quad\forall k\in\mathbb{Z}_{+}.
\endaligned
\end{equation}

\noindent Owing to $\rank\left(\bm{F}_{1}^{\tp}\bm{G}\right)=m$, we denote $\bm{G}=\left[\bm{G}_{1}~\bm{G}_{2}\right]$ with $\bm{G}_{1}\in\mathbb{R}^{p\times m}$ and $\bm{G}_{2}\in\mathbb{R}^{p\times(q-m)}$, where we consider $\rank\left(\bm{F}_{1}^{\tp}\bm{G}_{1}\right)=m$ without any loss of generality. We correspondingly denote $\bm{K}=\left[\bm{K}_{1}^{\tp}~\bm{K}_{2}^{\tp}\right]^{\tp}$ with $\bm{K}_{1}\in\mathbb{R}^{m\times p}$ and $\bm{K}_{2}\in\mathbb{R}^{(q-m)\times p}$. From the condition (\ref{eq20}), we know that $\bm{F}_{1}^{\tp}\bm{G}\bm{K}\bm{H}_{1}$ is nonsingular, and thus can develop a nonsingular matrix $\bm{\Omega}=\left[\bm{\Omega}_{1}~\bm{\Omega}_{2}\right]\in\mathbb{R}^{q\times q}$ with $\bm{\Omega}_{1}\in\mathbb{R}^{q\times m}$ and $\bm{\Omega}_{2}\in\mathbb{R}^{q\times(q-m)}$ in the form of
\[\aligned
\bm{\Omega}_{1}
&=\bm{K}\bm{H}_{1}\left(\bm{F}_{1}^{\tp}\bm{G}\bm{K}\bm{H}_{1}\right)^{-1}=\begin{bmatrix}\bm{K}_{1}\bm{H}_{1}\left(\bm{F}_{1}^{\tp}\bm{G}\bm{K}\bm{H}_{1}\right)^{-1}\\
\bm{K}_{2}\bm{H}_{1}\left(\bm{F}_{1}^{\tp}\bm{G}\bm{K}\bm{H}_{1}\right)^{-1}\end{bmatrix}\\
\bm{\Omega}_{2}
&=\begin{bmatrix}-\left(\bm{F}_{1}^{\tp}\bm{G}_{1}\right)^{-1}\bm{F}_{1}^{\tp}\bm{G}_{2}\\
I\end{bmatrix}
\endaligned
\]

\noindent of which the inverse matrix is given by $\bm{\Omega}^{-1}\triangleq\bm{\Xi}=\left[\bm{\Xi}_{1}^{\tp}~\bm{\Xi}_{2}^{\tp}\right]^{\tp}$ with $\bm{\Xi}_{1}\in\mathbb{R}^{m\times q}$ and $\bm{\Xi}_{2}\in\mathbb{R}^{(q-m)\times q}$ satisfying
\[\aligned
\bm{\Xi}_{1}
&=\bm{F}_{1}^{\tp}\bm{G}=\begin{bmatrix}\bm{F}_{1}^{\tp}\bm{G}_{1}&\bm{F}_{1}^{\tp}\bm{G}_{2}\end{bmatrix}\\
\bm{\Xi}_{2}
&=\begin{bmatrix}
\left[-\bm{K}_{2}\bm{H}_{1}\left(\bm{F}_{1}^{\tp}\bm{G}\bm{K}\bm{H}_{1}\right)^{-1}\bm{F}_{1}^{\tp}\bm{G}_{1}\right]^{\tp}\\
\left[I-\bm{K}_{2}\bm{H}_{1}\left(\bm{F}_{1}^{\tp}\bm{G}\bm{K}\bm{H}_{1}\right)^{-1}\bm{F}_{1}^{\tp}\bm{G}_{2}\right]^{\tp}\end{bmatrix}^{\tp}.
\endaligned\]

\noindent With $\bm{Y}_{d}\in\mathcal{Y}_{T}$ and $\sn\bm{G}=\sn\bm{H}_{1}$, $\bm{Y}_{d}=\bm{H}_{1}\bm{\eta}$ holds for some $\bm{\eta}\in\mathbb{R}^{m}$, and thus we can leverage Lemmas \ref{lem01} and \ref{lem02} to derive
\[
\aligned
\bm{\Xi}_{1}\bm{K}\bm{Y}_{d}
&=\bm{F}_{1}^{\tp}\bm{G}\bm{K}\bm{Y}_{d}\\
&=\bm{F}_{1}^{\tp}\bm{G}\bm{K}\left(\bm{H}_{1}\bm{F}_{1}^{\tp}+\bm{H}_{2}\bm{F}_{2}^{\tp}\right)\bm{Y}_{d}\\
&=\bm{F}_{1}^{\tp}\bm{G}\bm{K}\bm{H}_{1}\bm{F}_{1}^{\tp}\bm{Y}_{d}
\endaligned
\]

\noindent and
\[
\aligned
\bm{\Xi}_{2}\bm{K}\bm{Y}_{d}
&=\bm{K}_{2}\bm{H}_{1}\bm{\eta}
-\bm{K}_{2}\bm{H}_{1}\left(\bm{F}_{1}^{\tp}\bm{G}\bm{K}\bm{H}_{1}\right)^{-1}\bm{F}_{1}^{\tp}\bm{G}_{1}\bm{K}_{1}\bm{H}_{1}\bm{\eta}\\
&~~~-\bm{K}_{2}\bm{H}_{1}\left(\bm{F}_{1}^{\tp}\bm{G}\bm{K}\bm{H}_{1}\right)^{-1}\bm{F}_{1}^{\tp}\bm{G}_{2}\bm{K}_{2}\bm{H}_{1}\bm{\eta}\\
&=\bm{K}_{2}\bm{H}_{1}\bm{\eta}
-\bm{K}_{2}\bm{H}_{1}\big[\left(\bm{F}_{1}^{\tp}\bm{G}\bm{K}\bm{H}_{1}\right)^{-1}\bm{F}_{1}^{\tp}\left(\bm{G}_{1}\bm{K}_{1}\right.\\
&~~~\left.+\bm{G}_{2}\bm{K}_{2}\right)\bm{H}_{1}\big]\bm{\eta}\\
&=0
\endaligned
\]

\noindent the combination of which yields
\begin{equation}\label{eq30}
\bm{\Omega}^{-1}\bm{K}\bm{Y}_{d}
=\begin{bmatrix}
\bm{\Xi}_{1}\bm{K}\bm{Y}_{d}\\
\bm{\Xi}_{2}\bm{K}\bm{Y}_{d}
\end{bmatrix}
=\begin{bmatrix}
\bm{F}_{1}^{\tp}\bm{G}\bm{K}\bm{H}_{1}\bm{F}_{1}^{\tp}\bm{Y}_{d}\\
0
\end{bmatrix}.
\end{equation}

\noindent Similarly, we can validate
\begin{equation}\label{eq31}
\bm{\Omega}^{-1}\bm{K}\bm{H}_{1}
=\begin{bmatrix}
\bm{F}_{1}^{\tp}\bm{G}\bm{K}\bm{H}_{1}\\
0
\end{bmatrix},\quad
\bm{F}_{1}^{\tp}\bm{G}\bm{\Omega}
=\begin{bmatrix}
I&0
\end{bmatrix}.
\end{equation}

\noindent By applying (\ref{eq30}) and (\ref{eq31}), we implement a nonsingular linear transformation of (\ref{eq29}) given by
\[
\bm{\Omega}^{-1}\bm{U}_{k}
\triangleq\bm{U}_{k}^{\ast}
=\begin{bmatrix}
\bm{U}_{1,k}^{\ast}\\
\bm{U}_{2,k}^{\ast}
\end{bmatrix}~\hbox{with}~
\left\{\aligned
\bm{U}_{1,k}^{\ast}&=\Xi_{1}\bm{U}_{k}\in\mathbb{R}^{m}\\
\bm{U}_{2,k}^{\ast}&=\Xi_{2}\bm{U}_{k}\in\mathbb{R}^{q-m}
\endaligned,\forall k\in\mathbb{Z}_{+}\right.
\]

\noindent and can obtain
\[\aligned
\bm{U}_{k+1}^{\ast}
&=\bm{\Omega}^{-1}\left(I-\bm{K}\bm{H}_{1}\bm{F}_{1}^{\tp}\bm{G}\right)\bm{\Omega}\bm{U}_{k}^{\ast}+\bm{\Omega}^{-1}\bm{K}\bm{Y}_{d}\\
&=\begin{bmatrix}
I-\bm{F}_{1}^{\tp}\bm{G}\bm{K}\bm{H}_{1}&0\\
0&I
\end{bmatrix}\bm{U}_{k}^{\ast}
+\begin{bmatrix}
\bm{F}_{1}^{\tp}\bm{G}\bm{K}\bm{H}_{1}\bm{F}_{1}^{\tp}\bm{Y}_{d}\\
0
\end{bmatrix}
\endaligned,~~\forall k\in\mathbb{Z}_{+}
\]

\noindent which can be decomposed into two separate subsystems of
\begin{equation}\label{eq32}
\bm{U}_{1,k+1}^{\ast}
=\left(I-\bm{F}_{1}^{\tp}\bm{G}\bm{K}\bm{H}_{1}\right)\bm{U}_{1,k}^{\ast}
+\bm{F}_{1}^{\tp}\bm{G}\bm{K}\bm{H}_{1}\bm{F}_{1}^{\tp}\bm{Y}_{d},\quad\forall k\in\mathbb{Z}_{+}
\end{equation}

\noindent and of
\begin{equation}\label{eq33}
\bm{U}_{2,k+1}^{\ast}
=\bm{U}_{2,k}^{\ast},\quad\forall k\in\mathbb{Z}_{+}.
\end{equation}

\noindent An immediate consequence of (\ref{eq33}) is
\begin{equation}\label{eq34}
\bm{U}_{2,k}^{\ast}
=\bm{U}_{2,0}^{\ast}
=\bm{\Xi}_{2}\bm{U}_{0},\quad\forall k\in\mathbb{Z}_{+}.
\end{equation}

\noindent By incorporating (\ref{eq20}) into (\ref{eq32}), we can verify
\begin{equation}\label{eq35}
\bm{U}_{1,\infty}^{\ast}
\triangleq\lim_{k\to\infty}\bm{U}_{1,k}^{\ast}
=\bm{F}_{1}^{\tp}\bm{Y}_{d}.
\end{equation}

\noindent With $\bm{U}_{k}=\bm{\Omega}\bm{U}_{k}^{\ast}$, the use of (\ref{eq34}) and (\ref{eq35}) ensures the existence of $\bm{U}_{\infty}$ such that
\begin{equation}\label{eq36}
\aligned
\bm{U}_{\infty}
&=\bm{\Omega}_{1}\bm{U}_{1,\infty}^{\ast}
+\bm{\Omega}_{2}\bm{U}_{2,0}^{\ast}\\
&=\bm{\Omega}_{1}\bm{F}_{1}^{\tp}\bm{Y}_{d}
+\bm{\Omega}_{2}\bm{\Xi}_{2}\bm{U}_{0}\\
&=\left[I-\bm{K}\bm{H}_{1}\left(\bm{F}_{1}^{\tp}\bm{G}\bm{K}\bm{H}_{1}\right)^{-1}\bm{F}_{1}^{\tp}\bm{G}\right]\bm{U}_{0}\\
&~~~+\bm{K}\bm{H}_{1}\left(\bm{F}_{1}^{\tp}\bm{G}\bm{K}\bm{H}_{1}\right)^{-1}\bm{F}_{1}^{\tp}\bm{Y}_{d}
\endaligned
\end{equation}

\noindent where we also insert the following fact:
\[
\bm{\Omega}_{2}\bm{\Xi}_{2}=I-\bm{\Omega}_{1}\bm{\Xi}_{1}=I-\bm{K}\bm{H}_{1}\left(\bm{F}_{1}^{\tp}\bm{G}\bm{K}\bm{H}_{1}\right)^{-1}\bm{F}_{1}^{\tp}\bm{G}.
\]

\noindent From (\ref{eq25}) and (\ref{eq36}), $\bm{U}_{\infty}\in\mathcal{U}_{\mathrm{ILC}}(\bm{Y}_{d})$ holds for any $\bm{Y}_{d}\in\mathcal{Y}_{T}$.

{\it Step (ii): We prove (\ref{eq26}), and the equivalence between (\ref{eq27}) and $\rank\left(\bm{G}\right)=q$.}

``$\mathcal{U}_{\mathrm{ILC}}(\bm{Y}_{d})\subseteq\mathcal{U}_{d}(\bm{Y}_{d})$:'' For any $\bm{\chi}\in\mathcal{U}_{\mathrm{ILC}}(\bm{Y}_{d})$, there exists some $\bm{U}_{0}\in\mathbb{R}^{q}$ such that
\begin{equation*}\label{}
\aligned
\bm{\chi}
&=\left[I-\bm{K}\bm{H}_{1}\left(\bm{F}_{1}^{\tp}\bm{G}\bm{K}\bm{H}_{1}\right)^{-1}\bm{F}_{1}^{\tp}\bm{G}\right]\bm{U}_{0}\\
&~~~+\bm{K}\bm{H}_{1}\left(\bm{F}_{1}^{\tp}\bm{G}\bm{K}\bm{H}_{1}\right)^{-1}\bm{F}_{1}^{\tp}\bm{Y}_{d}
\endaligned
\end{equation*}

\noindent which, together with the results of Lemmas \ref{lem01} and \ref{lem02}, leads to
\begin{equation*}\label{}
\aligned
\bm{G}\bm{\chi}
&=\left(\bm{H}_{1}\bm{F}_{1}^{\tp}+\bm{H}_{2}\bm{F}_{2}^{\tp}\right)\bm{G}\bm{\chi}\\
&=\bm{H}_{1}\bm{F}_{1}^{\tp}\bm{G}\left[I-\bm{K}\bm{H}_{1}\left(\bm{F}_{1}^{\tp}\bm{G}\bm{K}\bm{H}_{1}\right)^{-1}\bm{F}_{1}^{\tp}\bm{G}\right]\bm{U}_{0}\\
&~~~+\bm{H}_{1}\bm{F}_{1}^{\tp}\bm{G}\bm{K}\bm{H}_{1}\left(\bm{F}_{1}^{\tp}\bm{G}\bm{K}\bm{H}_{1}\right)^{-1}\bm{F}_{1}^{\tp}\bm{Y}_{d}\\
&=\bm{H}_{1}\bm{F}_{1}^{\tp}\bm{Y}_{d}\\
&=\bm{Y}_{d}
\endaligned
\end{equation*}

\noindent namely, $\bm{\chi}\in\mathcal{U}_{d}(\bm{Y}_{d})$. This implies $\mathcal{U}_{\mathrm{ILC}}(\bm{Y}_{d})\subseteq\mathcal{U}_{d}(\bm{Y}_{d})$.

``$\mathcal{U}_{\mathrm{ILC}}(\bm{Y}_{d})\supseteq\mathcal{U}_{d}(\bm{Y}_{d})$:'' For any $\bm{\chi}\in\mathcal{U}_{d}(\bm{Y}_{d})$, we have $\bm{Y}_{d}=\bm{G}\bm{\chi}$ by (\ref{eq24}). Let us define
\[
\bm{U}_{0}=\bm{\chi}-\bm{K}\bm{H}_{1}\bm{\eta},\quad\forall\bm{\eta}\in\mathbb{R}^{m}.
\]

\noindent Since $\bm{K}\bm{H}_{1}$ is a matrix of full-column rank under the condition (\ref{eq20}), there exists $\bm{U}_{0}\in\mathbb{R}^{q}$ as defined above such that
\begin{equation*}\label{}
\aligned
&\left[I-\bm{K}\bm{H}_{1}\left(\bm{F}_{1}^{\tp}\bm{G}\bm{K}\bm{H}_{1}\right)^{-1}\bm{F}_{1}^{\tp}\bm{G}\right]
\left(\bm{\chi}-\bm{U}_{0}\right)\\
&~~~=\left[I-\bm{K}\bm{H}_{1}\left(\bm{F}_{1}^{\tp}\bm{G}\bm{K}\bm{H}_{1}\right)^{-1}\bm{F}_{1}^{\tp}\bm{G}\right]
\bm{K}\bm{H}_{1}\bm{\eta}\\
&~~~=0
\endaligned
\end{equation*}

\noindent which is equivalent to
\begin{equation*}\label{}
\aligned
\bm{\chi}
&=\left[I-\bm{K}\bm{H}_{1}\left(\bm{F}_{1}^{\tp}\bm{G}\bm{K}\bm{H}_{1}\right)^{-1}\bm{F}_{1}^{\tp}\bm{G}\right]\bm{U}_{0}\\
&~~~+\bm{K}\bm{H}_{1}\left(\bm{F}_{1}^{\tp}\bm{G}\bm{K}\bm{H}_{1}\right)^{-1}\bm{F}_{1}^{\tp}\bm{G}\bm{\chi}\\
&=\left[I-\bm{K}\bm{H}_{1}\left(\bm{F}_{1}^{\tp}\bm{G}\bm{K}\bm{H}_{1}\right)^{-1}\bm{F}_{1}^{\tp}\bm{G}\right]\bm{U}_{0}\\
&~~~+\bm{K}\bm{H}_{1}\left(\bm{F}_{1}^{\tp}\bm{G}\bm{K}\bm{H}_{1}\right)^{-1}\bm{F}_{1}^{\tp}\bm{Y}_{d}
\endaligned
\end{equation*}

\noindent namely, $\bm{\chi}\in\mathcal{U}_{\mathrm{ILC}}(\bm{Y}_{d})$. Thus, we can get $\mathcal{U}_{\mathrm{ILC}}(\bm{Y}_{d})\supseteq\mathcal{U}_{d}(\bm{Y}_{d})$.

Thanks to $\mathcal{U}_{\mathrm{ILC}}(\bm{Y}_{d})\subseteq\mathcal{U}_{d}(\bm{Y}_{d})$ and $\mathcal{U}_{\mathrm{ILC}}(\bm{Y}_{d})\supseteq\mathcal{U}_{d}(\bm{Y}_{d})$, we can develop (\ref{eq26}). From Theorem \ref{thm01}, $\mathcal{Y}_{T}=\mathcal{Y}_{R}$ is equivalent to $\rank\left(\bm{G}\right)=q$. Thus, $\mathcal{Y}_{T}=\mathcal{Y}_{R}$ yields $\rank\left(\bm{F}_{1}^{\tp}\bm{G}\right)=\rank\left(\bm{G}\right)=q$, and consequently $\bm{F}_{1}^{\tp}\bm{G}$ and $\bm{K}\bm{H}_{1}$ are nonsingular, owing to which (\ref{eq25}) collapses into
\begin{equation}\label{eq37}
\mathcal{U}_{\mathrm{ILC}}(\bm{Y}_{d})
=\left\{\left(\bm{F}_{1}^{\tp}\bm{G}\right)^{-1}\bm{F}_{1}^{\tp}\bm{Y}_{d}\right\},\quad\forall\bm{Y}_{d}\in\mathcal{Y}_{T}.
\end{equation}

\noindent Since we have $\left(\bm{F}_{1}^{\tp}\bm{G}\right)^{-1}\bm{F}_{1}^{\tp}=\left(\bm{G}^{\tp}\bm{G}\right)^{-1}\bm{G}^{\tp}$ by $\rank\left(\bm{F}_{1}^{\tp}\bm{G}\right)=\rank\left(\bm{G}\right)=q$, (\ref{eq27}) follows from (\ref{eq26}) and (\ref{eq37}). On the contrary, if (\ref{eq27}) holds, then for any $\bm{Y}_{d}\in\mathcal{Y}_{T}$, there exists a unique $\bm{U}_{d}$ to guarantee (\ref{eq1}) (namely, $\bm{Y}_{d}\in\mathcal{Y}_{R}$), based on which $\rank\left(\bm{G}\right)=q$ follows from Lemma \ref{lem03}, and equivalently $\mathcal{Y}_{T}=\mathcal{Y}_{R}$ holds. That is, we have the equivalence between (\ref{eq27}) and $\mathcal{Y}_{T}=\mathcal{Y}_{R}$.
\end{proof}

\section{Proof of Theorem \ref{t1}}

We first present a useful lemma to provide the general form of the least squares solutions to the LAE \eqref{eq1} for any $\bm{Y}_{d}\notin\mathcal{Y}_{T}$.

\begin{lem}\label{lem08}
For the LAE (\ref{eq1}), the set $\overline{\mathcal{U}}_{d}\left(\bm{Y}_{d}\right)$ of all its least squares solutions for any $\bm{Y}_{d}\notin\mathcal{Y}_{T}$ can be described by
\begin{equation}\label{e9}
\aligned
\overline{\mathcal{U}}_{d}\left(\bm{Y}_{d}\right)
=\bigg\{\bm{G}^{\{1,3\}}\bm{Y}_d+\left(I-\bm{G}^{\{1,3\}}\bm{G}\right)\bm{Z}\Big| \bm{G}^{\{1,3\}}&\in\widetilde{\bm{G}},\\
\bm{Z}\in\mathbb{R}^{q}\bigg\},\quad\forall\bm{Y}_{d}&\notin\mathcal{Y}_{T}
\endaligned
\end{equation}

\noindent where the matrix set $\widetilde{\bm{G}}$ is given by
\begin{equation}\label{e8}
\aligned
\widetilde{\bm{G}}
=\bigg\{\bm{G}^{\{1,3\}}\in\mathbb{R}^{q\times p}
\Big|\bm{G}\bm{G}^{\{1,3\}}\bm{G}&=\bm{G},\\
\left(\bm{G}\bm{G}^{\{1,3\}}\right)^{\tp}&=\bm{G}\bm{G}^{\{1,3\}}\bigg\}.
\endaligned
\end{equation}
\end{lem}

\begin{proof}
For any $\bm{G}^{\{1,3\}}\in\widetilde{\bm{G}}$, we notice (\ref{e8}) and can derive
\[\aligned
&\left(\bm{G}\bm{G}^{\{1,3\}}\bm{Y}_d-\bm{Y}_d\right)^{\tp}\left(\bm{G}\widetilde{\bm{U}}-\bm{G}\bm{G}^{\{1,3\}}\bm{Y}_d\right)\\
&~~~~~~~~~~~~~~=\left(\bm{Y}_d^{\tp}\bm{G}\bm{G}^{\{1,3\}}-\bm{Y}_d^{\tp}\right)\left(\bm{G}\widetilde{\bm{U}}-\bm{G}\bm{G}^{\{1,3\}}\bm{Y}_d\right)\\
&~~~~~~~~~~~~~~=\bm{Y}_d^{\tp}\bm{G}\bm{G}^{\{1,3\}}\bm{G}\widetilde{\bm{U}}
-\bm{Y}_d^{\tp}\left(\bm{G}\bm{G}^{\{1,3\}}\right)^{2}\bm{Y}_d\\
&~~~~~~~~~~~~~~~~~-\bm{Y}_d^{\tp}\bm{G}\widetilde{\bm{U}}
+\bm{Y}_d^{\tp}\bm{G}\bm{G}^{\{1,3\}}\bm{Y}_d\\
&~~~~~~~~~~~~~~=0,\quad\forall\widetilde{\bm{U}}\in\mathbb{R}^{q},\forall\bm{Y}_{d}\in\mathbb{R}^{p}.
\endaligned
\]

\noindent Namely, $\bm{G}\bm{G}^{\{1,3\}}\bm{Y}_d-\bm{Y}_d$ and $\bm{G}\widetilde{\bm{U}}-\bm{G}\bm{G}^{\{1,3\}}\bm{Y}_d$ are orthogonal for any $\widetilde{\bm{U}}\in\mathbb{R}^{q}$ and any $\bm{Y}_{d}\in\mathbb{R}^{p}$, with which we can deduce
\[\aligned
\left\|\bm{G}\widetilde{\bm{U}}-\bm{Y}_d\right\|_{2}
&=\left\|\left(\bm{G}\widetilde{\bm{U}}-\bm{G}\bm{G}^{\{1,3\}}\bm{Y}_d\right)
+\left(\bm{G}\bm{G}^{\{1,3\}}\bm{Y}_d-\bm{Y}_d\right)\right\|_2\\
&=\left\|\bm{G}\widetilde{\bm{U}}-\bm{G}\bm{G}^{\{1,3\}}\bm{Y}_d\right\|_2
+\left\|\bm{G}\bm{G}^{\{1,3\}}\bm{Y}_d-\bm{Y}_d\right\|_2\\
&\geq\left\|\bm{G}\bm{G}^{\{1,3\}}\bm{Y}_d-\bm{Y}_d\right\|_2,\quad\forall\widetilde{\bm{U}}\in\mathbb{R}^{q},\forall\bm{Y}_{d}\in\mathbb{R}^{p}.
\endaligned\]

\noindent This together with (\ref{e1}) ensures that $\overline{\bm{U}}_{d}\in\overline{\mathcal{U}}_{d}\left(\bm{Y}_{d}\right)$ holds if and only if $\overline{\bm{U}}_{d}$ satisfies $\left\|\bm{G}\overline{\bm{U}}_d-\bm{G}\bm{G}^{\{1,3\}}\bm{Y}_d\right\|_2=0$, or equivalently, $\overline{\bm{U}}_{d}\in\overline{\mathcal{U}}_{d}\left(\bm{Y}_{d}\right)$ is exactly the solution to the resulting LAE as
\begin{equation}\label{e010}
\bm{G}\left(\overline{\bm{U}}_d-\bm{G}^{\{1,3\}}\bm{Y}_d\right)=0.
\end{equation}

\noindent Thanks to $\nl\bm{G}=\nl\left(\bm{G}^{\{1,3\}}\bm{G}\right)$ and $\left(\bm{G}^{\{1,3\}}\bm{G}\right)^{2}=\bm{G}^{\{1,3\}}\bm{G}$, we can leverage the properties of the idempotent matrices (see, e.g., Theorem 1 of \cite[Subchapter 5.8]{lt:85}) to obtain
\begin{equation}\label{e011}
\nl\bm{G}=\sn\left(I-\bm{G}^{\{1,3\}}\bm{G}\right).
\end{equation}

\noindent For the LAE (\ref{e010}), since $\overline{\bm{U}}_d=\bm{G}^{\{1,3\}}\bm{Y}_d$ is a particular solution, we resort to (\ref{e011}) and can conclude that any solution takes the general form of (see, e.g., Theorem 7 of \cite[Subchapter 12.6]{lt:85})
\begin{equation*}\label{}
\overline{\bm{U}}_d
=\bm{G}^{\{1,3\}}\bm{Y}_d+\left(I-\bm{G}^{\{1,3\}}\bm{G}\right)\bm{Z},\quad\forall\bm{Z}\in\mathbb{R}^q.
\end{equation*}

\noindent Namely, the description of $\overline{\mathcal{U}}_{d}\left(\bm{Y}_{d}\right)$ in (\ref{e9}) is identical to that of (\ref{e1}). The proof of this lemma is complete.
\end{proof}

Based on Lemma \ref{lem08}, we are in position to present the proof of Theorem \ref{t1} as follows.

\begin{proof}[Proof of Theorem \ref{t1}]
We next adopt two separate steps to develop this theorem.

{\it Step (i): We prove that for any $\bm{Y}_{d}\notin\mathcal{Y}_{T}$, $\bm{Y}_{\infty}$ and $\bm{U}_{\infty}$ exist and satisfy $\bm{Y}_{\infty}=\bm{H}_1\bm{F}_1^{\tp}\bm{Y}_d$ and $\bm{U}_{\infty}\in\overline{\mathcal{U}}_{\mathrm{ILC}}(\bm{Y}_{d})$, respectively.}

With Theorem \ref{thm05}, we substitute (\ref{eq19}) into \eqref{eq13} and can obtain
\begin{equation}\label{e16}
\widehat{\bm{E}}_{k+1}^{C}
=\widehat{\bm{E}}_{k}^{C}-\bm{F}_{1}^{\tp}\bm{G}\bm{K}\bm{E}_k
=\left(I-\bm{F}_{1}^{\tp}\bm{G}\widehat{\bm{K}}\right)\widehat{\bm{E}}_{k}^{C},\quad\forall k\in\mathbb{Z}_{+}
\end{equation}

\noindent where we also use $\bm{\Psi}_{k}=-\Delta\bm{U}_{k}$, $\forall k\in\mathbb{Z}_{+}$, $\bm{K}=\bm{\widehat{K}}\bm{F}^{\tp}_1$ and $\widehat{\bm{E}}_{k}^{C}=\bm{F}_1^{\tp}\bm{E}_k$. With the condition (\ref{eq21}), $\lim_{k\to\infty}\widehat{\bm{E}}_{k}^{C}=0$ follows directly from (\ref{e16}). This, together with \eqref{eq18} and Lemma \ref{lem01}, yields
\begin{equation}\label{e5}
\aligned
\lim_{k\to\infty}\bm{E}_k
&=\bm{H}\lim_{k\to\infty}\widehat{\bm{E}}_{k}
=\left[\bm{H}_1~ \bm{H}_2\right]\left[\begin{matrix}0\\
\bm{F}_2^{\tp}\bm{E}_0\end{matrix}\right]
=\bm{H}_2\bm{F}_2^{\tp}\left(\bm{Y}_d-\bm{G}\bm{U}_0\right)\\
&=\left(I-\bm{H}_1\bm{F}_1^{\tp}\right)\bm{Y}_d.
\endaligned
\end{equation}

\noindent As a consequence of (\ref{e5}), $\bm{Y}_\infty=\lim_{k\to\infty}\left(\bm{Y}_d-\bm{E}_k\right)=\bm{H}_1\bm{F}_1^{\tp}\bm{Y}_d$ is immediate, namely, \eqref{e2} holds.

From \eqref{eq19}, we can also deduce
\begin{equation}\label{e6}
\aligned
\bm{U}_{k+1}
&=\bm{U}_{k}+\widehat{\bm{K}}\bm{F}_1^{\tp}\bm{E}_{k}\\
&=\bm{U}_{k}+\widehat{\bm{K}}\bm{F}_1^{\tp}\left(\bm{Y}_d-\bm{G}\bm{U}_{k}\right)\\
&=\left(I-\widehat{\bm{K}}\bm{F}_1^{\tp}\bm{G}\right)\bm{U}_{k}+\widehat{\bm{K}}\bm{F}_1^{\tp}\bm{Y}_d,\quad\forall k\in\mathbb{Z}_{+}.
\endaligned
\end{equation}

\noindent Since the condition \eqref{eq21} ensures the nonsingularity of $\bm{F}_1^{\tp}\bm{G}\widehat{\bm{K}}$, we can construct a nonsingular matrix $\overline{\bm{\Omega}}=\left[\overline{\bm{\Omega}}_{1}~\overline{\bm{\Omega}}_2\right]\in\mathbb{R}^{q\times q}$ with $\overline{\bm{\Omega}}_{1}\in\mathbb{R}^{q\times m}$ and $\overline{\bm{\Omega}}_{2}\in\mathbb{R}^{q\times(q-m)}$ given by
\[\aligned
\overline{\bm{\Omega}}_{1}
&=\widehat{\bm{K}}\left(\bm{F}_{1}^{\tp}\bm{G}\widehat{\bm{K}}\right)^{-1}=\begin{bmatrix}\widehat{\bm{K}}_{1}\left(\bm{F}_{1}^{\tp}\bm{G}\widehat{\bm{K}}\right)^{-1}\\
\widehat{\bm{K}}_{2}\left(\bm{F}_{1}^{\tp}\bm{G}\widehat{\bm{K}}\right)^{-1}\end{bmatrix}\\
\overline{\bm{\Omega}}_{2}
&=\begin{bmatrix}-\left(\bm{F}_{1}^{\tp}\bm{G}_{1}\right)^{-1}\bm{F}_{1}^{\tp}\bm{G}_{2}\\
I\end{bmatrix}
\endaligned
\]

\noindent where $\rank\left(\bm{F}_{1}^{\tp}\bm{G}_{1}\right)=m$ is assumed as in the proof of Theorem \ref{thm08}. Then the inverse matrix of $\overline{\bm{\Omega}}$ can be given by $\overline{\bm{\Xi}}=\left[\overline{\bm{\Xi}}_{1}^{\tp}~\overline{\bm{\Xi}}_{2}^{\tp}\right]^{\tp}$ with $\overline{\bm{\Xi}}_{1}\in\mathbb{R}^{m\times q}$ and $\overline{\bm{\Xi}}_{2}\in\mathbb{R}^{(q-m)\times q}$ satisfying
\[\aligned
\overline{\bm{\Xi}}_{1}
&=\bm{F}_{1}^{\tp}\bm{G}=\begin{bmatrix}\bm{F}_{1}^{\tp}\bm{G}_{1}&\bm{F}_{1}^{\tp}\bm{G}_{2}\end{bmatrix}\\
\overline{\bm{\Xi}}_{2}
&=\begin{bmatrix}
\left[-\widehat{\bm{K}}_{2}\left(\bm{F}_{1}^{\tp}\bm{G}\widehat{\bm{K}}\right)^{-1}\bm{F}_{1}^{\tp}\bm{G}_{1}\right]^{\tp}\\
\left[I-\widehat{\bm{K}}_{2}\left(\bm{F}_{1}^{\tp}\bm{G}\widehat{\bm{K}}\right)^{-1}\bm{F}_{1}^{\tp}\bm{G}_{2}\right]^{\tp}\end{bmatrix}^{\tp}.
\endaligned\]

\noindent By taking a nonsingular linear transformation of \eqref{e6} with
\[
\overline{\bm{\Omega}}^{-1}\bm{U}_{k}
\triangleq\overline{\bm{U}}_{k}^{\ast}
=\begin{bmatrix}
\overline{\bm{U}}_{1,k}^{\ast}\\
\overline{\bm{U}}_{2,k}^{\ast}
\end{bmatrix}~\hbox{with}~
\left\{\aligned
\overline{\bm{U}}_{1,k}^{\ast}&=\overline{\Xi}_{1}\bm{U}_{k}\in\mathbb{R}^{m}\\
\overline{\bm{U}}_{2,k}^{\ast}&=\overline{\Xi}_{2}\bm{U}_{k}\in\mathbb{R}^{q-m}
\endaligned,\forall k\in\mathbb{Z}_{+}\right.
\]

\noindent we can derive
\begin{equation*}\label{}
\aligned
\overline{\bm{U}}_{k+1}^{\ast}
&=\overline{\bm{\Omega}}^{-1}\left(I-\widehat{\bm{K}}\bm{F}_{1}^{\tp}\bm{G}\right)\overline{\bm{\Omega}}\overline{\bm{U}}_{k}^{\ast}+\overline{\bm{\Omega}}^{-1}\widehat{\bm{K}}\bm{F}_{1}^{\tp}\bm{Y}_{d}\\
&=\begin{bmatrix}
I-\bm{F}_{1}^{\tp}\bm{G}\widehat{\bm{K}}&0\\
0&I
\end{bmatrix}\overline{\bm{U}}_{k}^{\ast}
+\begin{bmatrix}
\bm{F}_{1}^{\tp}\bm{G}\widehat{\bm{K}}\bm{F}_{1}^{\tp}\bm{Y}_{d}\\
0
\end{bmatrix}
\endaligned,~~\forall k\in\mathbb{Z}_{+}
\end{equation*}

\noindent which clearly consists of two separate subsystems as
\begin{equation}\label{e7}
\overline{\bm{U}}_{1,k+1}^{\ast}
=\left(I-\bm{F}_{1}^{\tp}\bm{G}\widehat{\bm{K}}\right)\overline{\bm{U}}_{1,k}^{\ast}
+\bm{F}_{1}^{\tp}\bm{G}\widehat{\bm{K}}\bm{F}_{1}^{\tp}\bm{Y}_{d},\quad\forall k\in\mathbb{Z}_{+}
\end{equation}

\noindent and
\[
\overline{\bm{U}}_{2,k+1}^{\ast}
=\overline{\bm{U}}_{2,k}^{\ast},\quad\forall k\in\mathbb{Z}_{+}
\]

\noindent or equivalently,
\begin{equation}\label{e07}
\overline{\bm{U}}_{2,k}^{\ast}
=\overline{\bm{U}}_{2,0}^{\ast}
=\overline{\bm{\Xi}}_{2}\bm{U}_{0},\quad\forall k\in\mathbb{Z}_{+}.
\end{equation}

\noindent For the subsystem (\ref{e7}), the condition (\ref{eq21}) ensures its solution to converge and satisfy
\[
\lim_{k\to\infty}\overline{\bm{U}}_{1,k}^{\ast}
=\bm{F}_{1}^{\tp}\bm{Y}_{d}.
\]

\noindent With this fact and based on (\ref{e07}), we consider the nonsingular linear transformation $\bm{U}_{k}=\overline{\bm{\Omega}}\overline{\bm{U}}_{k}^{\ast}$ and can arrive at the existence of $\bm{U}_{\infty}$. Further, we notice $\overline{\bm{\Omega}}_{2}\overline{\bm{\Xi}}_{2}=I-\overline{\bm{\Omega}}_{1}\overline{\bm{\Xi}}_{1}$ and can deduce
\[
\aligned
\bm{U}_{\infty}
&=\overline{\bm{\Omega}}_{1}\overline{\bm{U}}_{1,\infty}^{\ast}
+\overline{\bm{\Omega}}_{2}\overline{\bm{U}}_{2,0}^{\ast}\\
&=\overline{\bm{\Omega}}_{1}\bm{F}_{1}^{\tp}\bm{Y}_{d}
+\overline{\bm{\Omega}}_{2}\overline{\bm{\Xi}}_{2}\bm{U}_{0}\\
&=\left[I-\widehat{\bm{K}}\left(\bm{F}_{1}^{\tp}\bm{G}\widehat{\bm{K}}\right)^{-1}\bm{F}_{1}^{\tp}\bm{G}\right]\bm{U}_{0}
+\widehat{\bm{K}}\left(\bm{F}_{1}^{\tp}\bm{G}\widehat{\bm{K}}\right)^{-1}\bm{F}_{1}^{\tp}\bm{Y}_{d}
\endaligned
\]

\noindent which yields $\bm{U}_{\infty}\in\overline{\mathcal{U}}_{\mathrm{ILC}}(\bm{Y}_{d})$ for any $\bm{Y}_{d}\notin\mathcal{Y}_{T}$ based on \eqref{e3}.

{\it Step (ii): We prove that \eqref{e4} holds if and only if $\bm{H}_1^{\tp}\bm{H}_2=0$.} With Lemma \ref{lem08}, we only need to show that $\widehat{\bm{K}}\left(\bm{F}_{1}^{\tp}\bm{G}\widehat{\bm{K}}\right)^{-1}\bm{F}_{1}^{\tp}\in\widetilde{\bm{G}}$ if and only if $\bm{H}_1^{\tp}\bm{H}_2=0$.

{\it Sufficiency:} According to Lemma \ref{lem01}, we can arrive at
\begin{equation}\label{e012}
\aligned
\bm{G}\widehat{\bm{K}}\left(\bm{F}_{1}^{\tp}\bm{G}\widehat{\bm{K}}\right)^{-1}\bm{F}_{1}^{\tp}
&=\left(\bm{H}_1\bm{F}_1^{\tp}+\bm{H}_2\bm{F}_2^{\tp}\right)
\bm{G}\widehat{\bm{K}}\left(\bm{F}_{1}^{\tp}\bm{G}\widehat{\bm{K}}\right)^{-1}\bm{F}_{1}^{\tp}\\
&=\bm{H}_{1}\bm{F}_{1}^{\tp}
\endaligned
\end{equation}

\noindent which can be employed to get
\[
\bm{G}\widehat{\bm{K}}\left(\bm{F}_{1}^{\tp}\bm{G}\widehat{\bm{K}}\right)^{-1}\bm{F}_{1}^{\tp}\bm{G}
=\bm{H}_{1}\bm{F}_{1}^{\tp}\bm{G}
=\left(I-\bm{H}_{2}\bm{F}_{2}^{\tp}\right)\bm{G}
=\bm{G}.
\]

\noindent From (\ref{e012}), the use of $\bm{H}_1^{\tp}\bm{H}_2=0$ and $\bm{H}_1\bm{F}_1^{\tp}+\bm{H}_2\bm{F}_2^{\tp}=I$ yields
\[\aligned
\bm{G}\widehat{\bm{K}}\left(\bm{F}_{1}^{\tp}\bm{G}\widehat{\bm{K}}\right)^{-1}\bm{F}_{1}^{\tp}
&=\left(\bm{H}_1\bm{F}_1^{\tp}+\bm{H}_2\bm{F}_2^{\tp}\right)^{\tp}\bm{H}_{1}\bm{F}_{1}^{\tp}\\
&=\bm{F}_1\bm{H}_1^{\tp}\bm{H}_{1}\bm{F}_{1}^{\tp}\\
&=\left(\bm{G}\widehat{\bm{K}}\left(\bm{F}_{1}^{\tp}\bm{G}\widehat{\bm{K}}\right)^{-1}\bm{F}_{1}^{\tp}\right)^{\tp}.
\endaligned\]

\noindent Thus, we have $\widehat{\bm{K}}\left(\bm{F}_{1}^{\tp}\bm{G}\widehat{\bm{K}}\right)^{-1}\bm{F}_{1}^{\tp}\in\widetilde{\bm{G}}$ by (\ref{e8}) when $\bm{H}_1^{\tp}\bm{H}_2=0$.

{\it Necessity:} If $\widehat{\bm{K}}\left(\bm{F}_{1}^{\tp}\bm{G}\widehat{\bm{K}}\right)^{-1}\bm{F}_{1}^{\tp}\in\widetilde{\bm{G}}$, then by (\ref{e8}), it follows
\[
\left(\bm{G}\widehat{\bm{K}}\left(\bm{F}_{1}^{\tp}\bm{G}\widehat{\bm{K}}\right)^{-1}\bm{F}_{1}^{\tp}\right)^{\tp}=\bm{G}\widehat{\bm{K}}\left(\bm{F}_{1}^{\tp}\bm{G}\widehat{\bm{K}}\right)^{-1}\bm{F}_{1}^{\tp}
\]

\noindent and as a consequence of (\ref{e012}), it is equivalent to
\begin{equation}\label{e17}
\bm{F}_{1}\bm{H}_{1}^{\tp}=\bm{H}_{1}\bm{F}_{1}^{\tp}.
\end{equation}

\noindent Due to $\bm{F}_{1}^{\tp}\bm{H}_2=0$, we resort to (\ref{e17}) and can obtain
\[
\bm{F}_{1}\bm{H}_{1}^{\tp}\bm{H}_2=\bm{H}_{1}\bm{F}_{1}^{\tp}\bm{H}_2=0
\]

\noindent from which $\bm{H}_{1}^{\tp}\bm{H}_2=0$ follows immediately, thanks to the full-column rank property of $\bm{F}_{1}$.
\end{proof}

\section{Proof of Theorem \ref{t3}}

\begin{proof}
For any $\bm{Y}_{d}\in\mathcal{Y}_T$, we follow the proof of Theorem \ref{thm08} to obtain that (\ref{e013}) is equivalent to
\begin{equation}\label{e015}
\bm{U}_{1,k}^{\ast}
\left\{\aligned
&\neq\bm{F}_{1}^{\tp}\bm{Y}_{d}, &\forall k&\leq\nu_{1}-1\\
&=\bm{F}_{1}^{\tp}\bm{Y}_{d}, &\forall k&\geq\nu_{1}.
\endaligned
\right.
\end{equation}

\noindent Since we can rewrite (\ref{eq32}) as
\begin{equation*}\label{}
\bm{U}_{1,k+1}^{\ast}-\bm{F}_{1}^{\tp}\bm{Y}_{d}
=\left(I-\bm{F}_{1}^{\tp}\bm{G}\bm{K}\bm{H}_{1}\right)
\left(\bm{U}_{1,k}^{\ast}-\bm{F}_{1}^{\tp}\bm{Y}_{d}\right),\quad\forall k\in\mathbb{Z}_{+}
\end{equation*}

\noindent we consequently have
\begin{equation}\label{e016}
\bm{U}_{1,k}^{\ast}-\bm{F}_{1}^{\tp}\bm{Y}_{d}
=\left(I-\bm{F}_{1}^{\tp}\bm{G}\bm{K}\bm{H}_{1}\right)^{k}
\left(\bm{U}_{1,0}^{\ast}-\bm{F}_{1}^{\tp}\bm{Y}_{d}\right),\quad\forall k\in\mathbb{Z}_{+}.
\end{equation}

\noindent Clearly, we can achieve (\ref{e015}) for (\ref{e016}) if and only if (\ref{e12}) holds, which is equivalent to (\ref{e11}) by properties of nilpotent matrices.

In the same way, we can develop the equivalent relationships among (\ref{e014}), (\ref{e13}), and (\ref{e14}) by applying Theorem \ref{t1} and further exploiting its proof, for which the details are thus omitted.
\end{proof}


%
%
%
%
%

\end{document}